\documentclass[pra,aps,onecolumn,twoside,superscriptaddress]{revtex4}

\usepackage{amsmath,amsfonts,amssymb,caption,color,epsfig,graphics,graphicx,hyperref,latexsym,mathrsfs,revsymb,theorem,url,verbatim,enumerate,epstopdf,tikz,float,multirow,booktabs,appendix,scalefnt}
\hypersetup{colorlinks,linkcolor={blue},citecolor={red},urlcolor={blue}}
\usetikzlibrary{arrows, decorations.markings}

\tikzstyle{vecArrow} = [thick, decoration={markings,mark=at position
	1 with {\arrow[semithick]{open triangle 60}}},
double distance=1.4pt, shorten >= 5.5pt,
preaction = {decorate},
postaction = {draw,line width=1.4pt, white,shorten >= 4.5pt}]
\tikzstyle{innerWhite} = [semithick, white,line width=1.4pt, shorten >= 4.5pt]

\newtheorem{definition}{Definition}
\newtheorem{proposition}[definition]{Proposition}
\newtheorem{lemma}[definition]{Lemma}

\newtheorem{theorem}[definition]{Theorem}
\newtheorem{corollary}[definition]{Corollary}
\newtheorem{conjecture}[definition]{Conjecture}

\newtheorem{remark}[definition]{Remark}
\newtheorem{example}[definition]{Example}
\newtheorem{question}[definition]{Question}

\def\bcj{\begin{conjecture}}
	\def\ecj{\end{conjecture}}
\def\bcr{\begin{corollary}}
	\def\ecr{\end{corollary}}
\def\bd{\begin{definition}}
	\def\ed{\end{definition}}
\def\bea{\begin{eqnarray}}
\def\eea{\end{eqnarray}}
\def\bem{\begin{enumerate}}
	\def\eem{\end{enumerate}}
\def\bex{\begin{example}}
	\def\eex{\end{example}}
\def\bim{\begin{itemize}}
	\def\eim{\end{itemize}}
\def\bl{\begin{lemma}}
	\def\el{\end{lemma}}
\def\bma{\begin{bmatrix}}
	\def\ema{\end{bmatrix}}
\def\bpf{\begin{proof}}
	\def\epf{\end{proof}}
\def\bpp{\begin{proposition}}
	\def\epp{\end{proposition}}
\def\bqu{\begin{question}}
	\def\equ{\end{question}}
\def\br{\begin{remark}}
	\def\er{\end{remark}}
\def\bt{\begin{theorem}}
	\def\et{\end{theorem}}

\def\squareforqed{\hbox{\rlap{$\sqcap$}$\sqcup$}}
\def\qed{\ifmmode\squareforqed\else{\unskip\nobreak\hfil
		\penalty50\hskip1em\null\nobreak\hfil\squareforqed
		\parfillskip=0pt\finalhyphendemerits=0\endgraf}\fi}
\def\endenv{\ifmmode\;\else{\unskip\nobreak\hfil
		\penalty50\hskip1em\null\nobreak\hfil\;
		\parfillskip=0pt\finalhyphendemerits=0\endgraf}\fi}
\newenvironment{proof}{\noindent \textbf{{Proof.~} }}{\qed}
\def\Dbar{\leavevmode\lower.6ex\hbox to 0pt
	{\hskip-.23ex\accent"16\hss}D}
\makeatletter
\def\url@leostyle{%
	\@ifundefined{selectfont}{\def\UrlFont{\sf}}{\def\UrlFont{\small\ttfamily}}}
\makeatother
\def\bcj{\begin{conjecture}}
	\def\ecj{\end{conjecture}}
\def\bcr{\begin{corollary}}
	\def\ecr{\end{corollary}}
\def\bd{\begin{definition}}
	\def\ed{\end{definition}}
\def\bea{\begin{eqnarray}}
\def\eea{\end{eqnarray}}
\def\bem{\begin{enumerate}}
	\def\eem{\end{enumerate}}
\def\bex{\begin{example}}
	\def\eex{\end{example}}
\def\bim{\begin{itemize}}
	\def\eim{\end{itemize}}
\def\bl{\begin{lemma}}
	\def\el{\end{lemma}}
\def\bpf{\begin{proof}}
	\def\epf{\end{proof}}
\def\bpp{\begin{proposition}}
	\def\epp{\end{proposition}}
\def\bqu{\begin{question}}
	\def\equ{\end{question}}
\def\br{\begin{remark}}
	\def\er{\end{remark}}
\def\bt{\begin{theorem}}
	\def\et{\end{theorem}}

\def\btb{\begin{tabular}}
	\def\etb{\end{tabular}}

\newcommand{\nc}{\newcommand}

\nc{\bbA}{\mathbb{A}} \nc{\bbB}{\mathbb{B}} \nc{\bbC}{\mathbb{C}}
\nc{\bbD}{\mathbb{D}} \nc{\bbE}{\mathbb{E}} \nc{\bbF}{\mathbb{F}}
\nc{\bbG}{\mathbb{G}} \nc{\bbH}{\mathbb{H}} \nc{\bbI}{\mathbb{I}}
\nc{\bbJ}{\mathbb{J}} \nc{\bbK}{\mathbb{K}} \nc{\bbL}{\mathbb{L}}
\nc{\bbM}{\mathbb{M}} \nc{\bbN}{\mathbb{N}} \nc{\bbO}{\mathbb{O}}
\nc{\bbP}{\mathbb{P}} \nc{\bbQ}{\mathbb{Q}} \nc{\bbR}{\mathbb{R}}
\nc{\bbS}{\mathbb{S}} \nc{\bbT}{\mathbb{T}} \nc{\bbU}{\mathbb{U}}
\nc{\bbV}{\mathbb{V}} \nc{\bbW}{\mathbb{W}} \nc{\bbX}{\mathbb{X}}
\nc{\bbZ}{\mathbb{Z}}

\nc{\bA}{{\bf A}} \nc{\bB}{{\bf B}} \nc{\bC}{{\bf C}}
\nc{\bD}{{\bf D}} \nc{\bE}{{\bf E}} \nc{\bF}{{\bf F}}
\nc{\bG}{{\bf G}} \nc{\bH}{{\bf H}} \nc{\bI}{{\bf I}}
\nc{\bJ}{{\bf J}} \nc{\bK}{{\bf K}} \nc{\bL}{{\bf L}}
\nc{\bM}{{\bf M}} \nc{\bN}{{\bf N}} \nc{\bO}{{\bf O}}
\nc{\bP}{{\bf P}} \nc{\bQ}{{\bf Q}} \nc{\bR}{{\bf R}}
\nc{\bS}{{\bf S}} \nc{\bT}{{\bf T}} \nc{\bU}{{\bf U}}
\nc{\bV}{{\bf V}} \nc{\bW}{{\bf W}} \nc{\bX}{{\bf X}}
\nc{\bZ}{{\bf Z}} 
 \nc{\bv}{{\bf v}}
\nc{\ba}{{\bf a}} \nc{\be}{{\bf e}} \nc{\bu}{{\bf u}}
\nc{\brr}{{\bf r}}

\nc{\cA}{{\cal A}} \nc{\cB}{{\cal B}} \nc{\cC}{{\cal C}}
\nc{\cD}{{\cal D}} \nc{\cE}{{\cal E}} \nc{\cF}{{\cal F}}
\nc{\cG}{{\cal G}} \nc{\cH}{{\cal H}} \nc{\cI}{{\cal I}}
\nc{\cJ}{{\cal J}} \nc{\cK}{{\cal K}} \nc{\cL}{{\cal L}}
\nc{\cM}{{\cal M}} \nc{\cN}{{\cal N}} \nc{\cO}{{\cal O}}
\nc{\cP}{{\cal P}} \nc{\cQ}{{\cal Q}} \nc{\cR}{{\cal R}}
\nc{\cS}{{\cal S}} \nc{\cT}{{\cal T}} \nc{\cU}{{\cal U}}
\nc{\cV}{{\cal V}} \nc{\cW}{{\cal W}} \nc{\cX}{{\cal X}}
\nc{\cZ}{{\cal Z}}

\nc{\hA}{{\hat{A}}} \nc{\hB}{{\hat{B}}} \nc{\hC}{{\hat{C}}}
\nc{\hD}{{\hat{D}}} \nc{\hE}{{\hat{E}}} \nc{\hF}{{\hat{F}}}
\nc{\hG}{{\hat{G}}} \nc{\hH}{{\hat{H}}} \nc{\hI}{{\hat{I}}}
\nc{\hJ}{{\hat{J}}} \nc{\hK}{{\hat{K}}} \nc{\hL}{{\hat{L}}}
\nc{\hM}{{\hat{M}}} \nc{\hN}{{\hat{N}}} \nc{\hO}{{\hat{O}}}
\nc{\hP}{{\hat{P}}} \nc{\hR}{{\hat{R}}} \nc{\hS}{{\hat{S}}}
\nc{\hT}{{\hat{T}}} \nc{\hU}{{\hat{U}}} \nc{\hV}{{\hat{V}}}
\nc{\hW}{{\hat{W}}} \nc{\hX}{{\hat{X}}} \nc{\hZ}{{\hat{Z}}}

\nc{\hn}{{\hat{n}}}

\def\rank{\mathop{\rm rank}}

\def\tr{\mathop{\rm Tr}}

\newcommand{\bra}[1]{\langle#1|}
\newcommand{\ket}[1]{|#1\rangle}

\newcommand{\ketbra}[2]{|#1\rangle\!\langle#2|}
\newcommand{\braket}[2]{\langle#1|#2\rangle}

\newcommand{\fl}[2]{\lfloor\frac{#1}{#2}\rfloor}

\def\Dbar{\leavevmode\lower.6ex\hbox to 0pt
	{\hskip-.23ex\accent"16\hss}D}

\begin{document}

\title{Strong quantum nonlocality for unextendible product bases in heterogeneous systems}


	\author{Fei Shi}
\email[]{shifei@mail.ustc.edu.cn}
\affiliation{School of Cyber Security,
	University of Science and Technology of China, Hefei, 230026, People's Republic of China}

\author{Mao-Sheng Li}
\email[]{li.maosheng.math@gmail.com}
\affiliation{Department of Physics, Southern University of Science and Technology, Shenzhen 518055, China}
\affiliation{Department of Physics, University of Science and Technology of China, Hefei 230026, China}

\author{Lin Chen}
\email[]{linchen@buaa.edu.cn}
\affiliation{LMIB(Beihang University), Ministry of Education, and School of Mathematical Sciences, Beihang University, Beijing 100191, China}
\affiliation{International Research Institute for Multidisciplinary Science, Beihang University, Beijing 100191, China
}

\author{Xiande Zhang}
\email[]{Corresponding author: drzhangx@ustc.edu.cn}
\affiliation{School of Mathematical Sciences,
	University of Science and Technology of China, Hefei, 230026, People's Republic of China}

\begin{abstract}
	A set of multipartite orthogonal product states is strongly nonlocal if it is locally irreducible in every bipartition, which shows the phenomenon of strong quantum nonlocality without entanglement. It is known that unextendible product bases (UPBs) can show the phenomenon of quantum nonlocality without entanglement. Thus it is interesting to investigate the strong quantum nonlocality for UPBs. Most of the UPBs with the minimum size cannot demonstrate strong quantum nonlocality. In this paper, we construct a series of UPBs with different large sizes in
	$d_A\otimes d_B\otimes d_C$ and  $d_A\otimes d_B\otimes d_C\otimes d_D$ for $d_A, d_B, d_C, d_D\geq 3$, and we also show that these UPBs have strong quantum nonlocality, which  answers an open question given by Halder \emph{et al.} [\href{https://journals.aps.org/prl/abstract/10.1103/PhysRevLett.122.040403}{Phys. Rev. Lett. \textbf{122}, 040403 (2019)}] and Yuan \emph{et al.} [\href{https://journals.aps.org/pra/abstract/10.1103/PhysRevA.102.042228}{Phys. Rev. A \textbf{102}, 042228 (2020)}] for any possible three and four-partite systems.  Furthermore, we  propose an  entanglement-assisted protocol to locally discriminate  the UPB in $3\otimes 3\otimes 4$, and it consumes less entanglement resource than the teleportation-based protocol. Our results build the connection between strong quantum nonlocality and UPBs.
\end{abstract}

    	\maketitle

\vspace{-0.5cm}
~~~~~~~~~~\indent{\textbf{Keywords}}: strong quantum nonlocality, unextendible product bases, entanglement-assisted discrimination

\section{Introduction}\label{sec:int}

 It is known that a set of nonorthogonal states cannot be perfectly distinguished, and a set of mutually orthogonal states can be always distinguished by performing a global measurement \cite{computation2010}.  However, when the composite quantum system is distributed among several spatially separated parties, it is not always possible to distinguish the states by performing local operations and classical communications (LOCC).  These states are said to be locally indistinguishable. The local indistinguishability plays an important role in quantum data hiding \cite{terhal2001hiding,divincenzo2002quantum,eggeling2002hiding,Matthews2009Distinguishability} and quantum secret sharing \cite{Markham2008Graph,Hillery,Rahaman}. Bennett \emph{et al.} first showed the phenomenon of quantum nonlocality without entanglement by constructing a locally indistinguishable orthogonal product basis in a two qutrit system \cite{bennett1999quantum}.   Consequently, the nonlocality here (or the \emph{local distinguishability based nonlocality} that we call) is very different from the most
well-known form of quantum nonlocality also known as Bell
nonlocality \cite{horodecki2009quantum,brunner2014bell} which can only arise  from entangled states.  This leads us to explore this kind of nonlocality.  After that,  locally indistinguishable
 sets have been widely investigated \cite{1,2,3,4,5,6,7,8,9,10,11,12,13,14,15,16,17,18,li1}.  Unextendible product bases (UPBs) are special kinds of locally indistinguishable sets  \cite{bennett1999unextendible,de2004distinguishability}. A  UPB for a multipartite quantum system is an incomplete orthogonal product basis whose complementary subspace contains no product state.  The UPBs can be used to construct positive-partial-transpose (PPT) entangled states and Peres sets \cite{bennett1999unextendible,halder2019family}. It is also connected to bound entangled states,  fermionic systems, Bell inequalities without quantum violation \cite{bennett1999unextendible,1,Tura2012four, Chen2014Unextendible, Augusiak2012tight,augusiak2011bell}. Most of the constructions for UPBs are about  UPBs with the minimum size \cite{bennett1999quantum,alon2001unextendible,1,Fen06,Joh13,Chen2013The}. It is interesting to investigate  UPBs  with large sizes in multipartite systems.

 Recently, Halder \emph{et al.}  proposed a strong form of nonlocality based on the concept of local irreducibility of  quantum states \cite{Halder2019Strong}. An orthogonal product set (OPS) is locally irreducible means that it is not possible to eliminate
 one or more states from the set by orthogonality-preserving local measurements.  Under this definition, a locally irreducible set must be a locally indistinguishable set,  while the converse is not true in general. Therefore, constructing locally irreducible set of orthogonal quantum states is an efficient way to show the local  distinguishability based  nonlocality.  An  OPS is said to be \emph{strongly nonlocal}  if it is locally irreducible in every
 bipartition.  For further study this kind of nonlocality,  it is interesting to investigate the locally irreducibility and the strong quantum nonlocality for OPSs.  Halder  \emph{et al.}  constructed two strongly nonlocal OPSs in $3\otimes 3\otimes 3$ and $4\otimes 4\otimes 4$ respectively, which shows the phenomenon of strong quantum nonlocality without entanglement \cite{Halder2019Strong}.     In addition, there were several results about strongly nonlocal OPSs. A strongly nonlocal OPS in $d\otimes d\otimes d$, $d\otimes d\otimes (d+1)$, $3\otimes 3\otimes 3\otimes 3$ and $4\otimes 4\otimes 4\otimes 4$  for $d\geq 3$ was given in \cite{yuan2020strong}. The authors in \cite{shi20211}  constructed a strongly nonlocal OPS in $d_A\otimes d_B\otimes d_C$,  $d_A\otimes d_B\otimes d_C\otimes d_D$, and  $d_A\otimes d_B\otimes d_C\otimes d_D\otimes d_E$ for $d_A,d_B,d_C,d_D,d_E\geq 3$.  Both of Refs. \cite{Halder2019Strong,yuan2020strong} propose an open question: whether one can find strongly nonlocal UPBs? Specially, Ref. \cite{yuan2020strong} indicates that most of the previous UPBs with the minimum size cannot be used for building strongly nonlocal UPBs. In \cite{shi2021}, the authors  show that
  a strongly nonlocal UPB with large size in $d\otimes d\otimes d$ exists for $d\geq 3$, where the UPB was constructed from \cite{Agrawal2019Genuinely}.  However, we still do not know whether there exist strongly nonlocal UPBs in multipartite systems with non-equal local dimensions.   So it is interesting to consider strongly nonlocal UPBs in general systems, like any possible  three and four-partite systems. Further, some strongly nonlocal orthogonal entangled sets were shown in \cite{2020Strong,li2}

When a set of orthogonal states is not locally distinguishable, entanglement can be used as a resource for distinguishability of such states. This is called the entanglement-assisted discrimination, which was first proposed by Cohen. In \cite{cohen2008understanding}, Cohen showed that the tile UPB in $3\otimes 3$ can be perfectly distinguished by using a two-qubit maximally entangled state. Since then, entanglement-assisted discrimination has attracted a lot of interest \cite{ghosh2001distinguishability,cohen2008understanding,bandyopadhyay2016entanglement,zhang2016entanglement,gungor2016entanglement,zhang2018local,Sumit2019Genuinely,zhang2020locally,Shi2020Unextendible,2020Strong}. Since a strongly nonlocal UPB is locally indistinguishable in any bipartition, a perfect discrimination of this set needs a resource state that must be entangled in all bipartitions. In case of the teleportation-based protocol \cite{Bennett1993Teleporting}, any set of orthogonal states in $m \otimes n$ ($m\leq n$) can be perfectly distinguished  by using an $m\otimes m$ maximally entangled state. Then the teleportation-based protocol can perfectly distinguish  the strongly nonlocal UPB in $3 \otimes 3\otimes 4$ by using $3\otimes 3$ maximally entangled states shared between any two pairs. Since entanglement is a costly resource in quantum information, it is important to find a protocol using cheaper resources.

In this paper, we focus on the construction of strongly nonlocal UPBs in any possible  three and four-partite systems. The construction of strongly nonlocal UPBs is more difficult than the construction of strongly nonlocal OPSs, since it is not easy to check that an OPS is a UPB usually. Thus, new method is required. By using the relation between OPSs and grid representations, we successfully construct a series of UPBs with different large sizes in three and four-partite systems. That is, for any $0\leq s\leq \fl{d_A-3}{2}$, there exists a UPB of size $d_Ad_Bd_C-8(s+1)$ in $d_A\otimes d_B\otimes d_C$  and a UPB of size $d_Ad_Bd_Cd_D-16(s+1)$ in $d_A\otimes d_B\otimes d_C\otimes d_D$ for $3\leq d_A\leq d_B\leq d_C\leq d_D$. We also show that these UPBs are strongly nonlocal by using the techniques from \cite{shi2021}. Further, we propose an entanglement-assisted protocol for local discrimination of the strongly nonlocal UPB in $3\otimes 3\otimes 4$, which consumes less entanglement resource than the teleportation-based protocol.

  The rest of this paper is organized as follows. In Sec.~\ref{sec:pre}, we introduce preliminary knowledge. In Sec.~\ref{sec:UPBd_Ad_Bd_C}, we construct tripartite UPBs and show that these UPBs are strongly nonlocal. In Theorem 3 of supplementary material \cite{Supplementary}, we show a series of strongly nonlocal UPBs in four-partite systems. In Sec.~\ref{sec:discrim}, we consider the entanglement-assisted discrimination for the strongly nonlocal UPB in $3\otimes  3\otimes 4$. Finally, we conclude in Sec.~\ref{sec:con}.

\section{Preliminary}\label{sec:pre}
In this section, we introduce the preliminary knowledge
and facts. Throughout this paper, we do not normalize states
and operators for simplicity,  and we consider only pure states and
positive operator-valued measure (POVM) measurements.
 For any positive integer $n\geq 1$, we denote $\bbZ_{n}$ as the set $\{0,1,\cdots,n-1\}$. Assume that $\{\ket{i}\}_{i\in\bbZ_m}$ is the computational basis of an $m$-dimensional Hilbert space. A bipartite state $\ket{\psi}$ in $m\otimes n$ can be expressed by
\begin{equation}
\ket{\psi}=\sum_{i\in\bbZ_m}\sum_{j\in\bbZ_n}a_{i,j}\ket{i}_A\ket{j}_B.
\end{equation}
Then $\ket{\psi}$ corresponds to an $m\times n$ matrix,
\begin{equation}
M=(a_{i,j})_{i\in\bbZ_m,j\in\bbZ_n}.
\end{equation}
Note that $\ket{\psi}$ is a product state if and only if $\rank(M)=1$.
Assume that $\ket{\psi_i}$ in $m\otimes n$ corresponds to an $m\times n$ matrix $M_i$ for $i=1,2$, then the inner product
\begin{equation}
\braket{\psi_1}{\psi_2}=\text{Tr}(M_1^{\dagger}M_2).
\end{equation}

An \emph{unextendible product basis} (UPB) for a multipartite quantum system is an incomplete orthogonal product basis whose complementary subspace contains no product state.  For example, the SHIFTS UPB in $2\otimes 2\otimes 2$ is as follows \cite{bennett1999unextendible},
\begin{equation*}
\begin{aligned}
\ket{\psi_0}&=\ket{0}_A\ket{1}_B\ket{+}_C,\quad
\ket{\psi_1}=\ket{1}_A\ket{+}_B\ket{0}_C,\\
\ket{\psi_2}&=\ket{+}_A\ket{0}_B\ket{1}_C,\quad
\ket{\psi_3}=\ket{-}_A\ket{-}_B\ket{-}_C,\\
\end{aligned}
\end{equation*}
where $\ket{\pm}=(\ket{0}\pm\ket{1})/\sqrt{2}$.

A local measurement performed to distinguish a set of multipartite mutually orthogonal states is called
an \emph{orthogonality-preserving local measurement}, if the
postmeasurement states keep being mutually orthogonal.
Specially, a measurement is \emph{trivial} if all the POVM elements are proportional to the identity operator.
In \cite{Halder2019Strong}, the authors proposed the concept of strong quantum nonlocality. An orthogonal product set (OPS) is said to be \emph{strongly nonlocal}  if it is locally irreducible in every
bipartition. Note that an OPS is
a \emph{locally irreducible set} means that  it is not possible to eliminate
one or more states from the set by orthogonality-preserving
local measurements.  There exists an efficient way to check whether an OPS is strongly nonlocal \cite{shi2021}. Assume that an OPS $\{\ket{\psi}\}\subset \otimes_{i=1}^n\cH_{A_i}$. Let $B_1=\{A_2A_3\ldots A_n\}$, $B_2=\{A_3\ldots A_nA_1\}, B_3=\{A_4\ldots A_nA_1A_2\}, \ldots, B_n=\{A_1\ldots A_{n-1}\}$. If the party $B_i$  can only perform a trivial orthogonality-preserving POVM for any $1\leq i\leq n$,  then the OPS  $\{\ket{\psi}\}$ is strongly nonlocal.

 In this paper, we show a series of  strongly nonlocal UPBs  in $d_A\otimes d_B\otimes d_C$ and $d_A\otimes d_B\otimes d_C\otimes d_D$ for any $d_A,d_B,d_C,d_D\geq 3$ respectively.  Without loss of generality, we always assume that $3\leq d_A\leq d_B\leq d_C\leq d_D$. In \cite{shi20211}, the authors gave a decomposition for the outermost layer of 3,4-dimensional hypercubes, and our construction of UPBs in this paper is inspired by this decomposition. Since any OPS in $2\otimes n$ is locally reducible \cite{1,halder2019family},  a  strongly nonlocal UPB in $d_A\otimes d_B \otimes d_C$ and  $d_A\otimes d_B\otimes d_C\otimes d_D$ must satisfy $d_A,d_B,d_C,d_D\geq 3$. Our construction achieves the minimum quantum system necessary for the existence of such a UPB.

Next, we introduce two basic lemmas from \cite{shi2021}, which are useful for showing strong quantum nonlocality. Let $\cH_n$ be an $n$-dimensional Hilbert space. Assume that the computational basis of $\cH_n$ is $\{|0\rangle,|1\rangle,\cdots, |n-1\rangle\}$. For any operator $E$ on $\cH_n$, we denote the matrix $E$ as the matrix representation of the operator $E$ under the computational basis. In general, we do not distinguish the operator $E$ and the  matrix $E$. Given any $n\times n$ matrix $E:=\sum_{i=0}^{n-1}\sum_{j=0}^{n-1} a_{i,j}|i\rangle\langle j|$, for $\mathcal{S},\mathcal{T}\subseteq \{|0\rangle,|1\rangle,\cdots, |n-1\rangle\}$, we define
\begin{equation*}
{}_\mathcal{S}E_{\mathcal{T}}:=\sum_{|s\rangle \in \mathcal{S}}\sum_{|t\rangle \in \mathcal{T}}a_{s,t} |s\rangle\langle t|.
\end{equation*}
It means that ${}_\mathcal{S}E_{\mathcal{T}}$ is a submatrix of $E$ with  row coordinates $\mathcal{S}$ and column coordinates $\mathcal{T}$.  In the case $\cS=\cT$, we denote $E_{\cS}:={}_{\cS}E_{\cS}$ for simplicity.  Given a set $\cS\subseteq\{|0\rangle,|1\rangle,\cdots, |n-1\rangle\}$, an orthogonal set $\{\ket{\psi_i}\}_{i\in\bbZ_s}$ is \emph{spanned} by $\cS$, if for any $i\in\bbZ_s$,  $\ket{\psi_i}$ is a linear combination of the states from $\cS$.

\begin{lemma}[Block Zeros Lemma \cite{shi2021}]\label{lem:zero}
	Let  an  $n\times n$ matrix $E=(a_{i,j})_{i,j\in\bbZ_n}$ be the matrix representation of an operator  $E=M^{\dagger}M$  under the basis  $\cB=\{\ket{0},\ket{1},\ldots,\ket{n-1}\}$. Given two nonempty disjoint subsets $\cS$ and $\cT$ of $\cB$, assume  that  $\{\ket{\psi_i}\}_{i\in\bbZ_s}$, $\{\ket{\phi_j}\}_{j\in\bbZ_t}$ are two orthogonal sets  spanned by $\cS$ and $\cT$ respectively, where $s=|\cS|,$ and $t=|\cT|.$  If  $\langle \psi_i| E| \phi_j\rangle =0$
	for any $i\in \mathbb{Z}_s$ and $j\in\mathbb{Z}_t$, then   ${}_\mathcal{S}E_{\mathcal{T}}=\mathbf{0}$  and  ${}_\mathcal{T}E_{\mathcal{S}}=\mathbf{0}$.
\end{lemma}

\begin{lemma}[Block Trivial  Lemma \cite{shi2021}]\label{lem:trivial}
	Let  an  $n\times n$ matrix $E=(a_{i,j})_{i,j\in\bbZ_n}$ be the matrix representation of an operator  $E=M^{\dagger}M$  under the basis  $\cB=\{\ket{0},\ket{1},\ldots,\ket{n-1}\}$. Given a nonempty  subset $\cS:=\{\ket{u_0},\ket{u_1},\ldots,\ket{u_{s-1}}\}$  of $\cB$, let $\{\ket{\psi_j} \}_{j\in\bbZ_s}$ be an orthogonal  set spanned by $\cS$.     Assume that $\langle \psi_i|E |\psi_j\rangle=0$ for any $i\neq j\in \mathbb{Z}_s$.  If there exists a state $|u_t\rangle \in\cS$,  such that $ {}_{\{|u_t\rangle\}}E_{\cS\setminus \{|u_t\rangle\}}=\mathbf{0}$  and $\langle u_t|\psi_j\rangle \neq 0$  for any $j\in \mathbb{Z}_s$, then  $E_{\cS}\propto \mathbb{I}_{\cS}.$
\end{lemma}

\section{Strongly nonlocal tripartite UPBs}\label{sec:UPBd_Ad_Bd_C}

In this section, we construct strongly nonlocal tripartite UPBs.  In Example~\ref{example:334}, we show a UPB in $3\otimes 3\otimes 4$. Then we generalize this UPB to the space $d_A\otimes d_B\otimes d_C$ in Proposition~\ref{pro:3d1d2}, and prove its strong quantum nonlocality in Proposition~\ref{pro:dAdBdCofthe}.  In fact, we show a series of strongly nonlocal UPBs of different sizes in $d_A\otimes d_B\otimes d_C$ in Theorem~\ref{thm:d1d2d3}.  Let $w_n=e^{\frac{2\pi\sqrt{-1}}{n}}$ be the $n$-th unit root, and let $\text{sum}(M)$ be the sum of all entries of the matrix $M$.

First, we consider a simple example in $3\otimes 3\otimes 4$.  Let
	\begin{equation}\label{OPB334}
\begin{aligned}
\cA_1&:=\{\ket{\psi_1(i,k)}=\ket{\xi_i}_A\ket{0}_B\ket{\eta_k}_C:\ (i,k)\in\bbZ_2\times \bbZ_3 \setminus \{(0,0)\} \},\\
\cA_2&:=\{\ket{\psi_2(i,j)}=\ket{\xi_i}_A\ket{\eta_j}_B\ket{3}_C:\ (i,j) \in\bbZ_2\times \bbZ_2 \setminus \{(0,0)\} \},\\			
\cA_3&:=\{\ket{\psi_3(j,k)}=\ket{2}_A\ket{\xi_j}_B\ket{\eta_k}_C:\ (j,k)\in\bbZ_2\times \bbZ_3 \setminus \{(0,0)\} \},\\
\cA_4&:=\{\ket{\psi_4}= |2\rangle_A|2\rangle_B|3\rangle_C\},\\
\cB_1&:=\{\ket{\phi_1(i,k)}=\ket{\eta_i}_A\ket{2}_B\ket{\xi_k}_C:\ (i,k)\in\bbZ_2\times \bbZ_3 \setminus \{(0,0)\}\},\\
\cB_2&:=\{\ket{\phi_2(i,j)}=\ket{\eta_i}_A\ket{\xi_j}_B\ket{0}_C:\ (i,j)\in\bbZ_2\times \bbZ_2 \setminus \{(0,0)\}\},\\
\cB_3&:=\{\ket{\phi_3(j,k)}=\ket{0}_A\ket{\eta_j}_B\ket{\xi_k}_C:\ (j,k)\in\bbZ_2\times \bbZ_3\setminus \{(0,0)\} \}, \\
\cB_4&:= \{\ket{\phi_4}= |0\rangle_A|0\rangle_B|0\rangle_C\},\\
\cF&:=\{\ket{\varphi(k)}=\ket{1}_A\ket{1}_B\ket{\beta_k}_C:\  k \in\bbZ_2 \setminus \{0\}\},
\end{aligned}
\end{equation}
\begin{equation*}
\ket{S}:=\left(\sum_{i=0}^{2}\ket{i}\right)_A\left(\sum_{j=0}^{2}\ket{j}\right)_B\left(\sum_{k=0}^{3}\ket{k}\right)_C,			
\end{equation*}
where $\ket{\eta_s}_X=\ket{0}_X+(-1)^s\ket{1}_X$, $\ket{\xi_s}_X=\ket{1}_X+(-1)^s\ket{2}_X$ for $s\in \bbZ_2$, $X\in\{A,B\}$, $\ket{\eta_s}_C=\sum_{t=0}^2w_3^{st}\ket{t}_C$,  $\ket{\xi_s}_C=\sum_{t=0}^2w_3^{st}\ket{t+1}_C$ for $s\in \bbZ_3$ and $\ket{\beta_s}_C=\ket{1}_C+(-1)^s\ket{2}_C$ for $s\in \bbZ_2$.

The state $\ket{S}$ is called a ``stopper" state. It is easy to see that $\cup_{i=1}^3 (\cA_i\cup \cB_i) \cup \cF \cup \{\ket{S}\}$ is an OPS, and $\cup_{i=1}^3 (\cA_i\cup \cB_i) \cup \cF \cup_{i=1}^3\{\ket{\psi_i(0,0)},\ket{\phi_i(0,0)}\}\cup \{\ket{\psi_4}\}\cup\{\ket{\phi_4}\}\cup \{\ket{\varphi(0)}\}$ is a complete orthogonal basis in $\bbC^3\otimes \bbC^3\otimes \bbC^4$.  The nine subsets $\cA_i,\cB_i$ $(i=1,2,3,4)$ and $\cF$ in $A|BC$ bipartition correspond to the nine blocks of $3\times 12$ grid in Fig.~\ref{fig:334}. For example, $\cA_1$ corresponds to the $2\times 3$ grid $\{(1,2)\times (00,01,02)\}$. Moreover,  $\cA_i$ is symmetrical to $\cB_i$ for $1\leq i\leq 4$. If we delete $\cA_4,\cB_4$, and add the ``stopper" state $\ket{S}$, we can obtain a UPB in $3\otimes 3\otimes 4$.

\begin{figure}[H]
	\centering
	\includegraphics[scale=0.8]{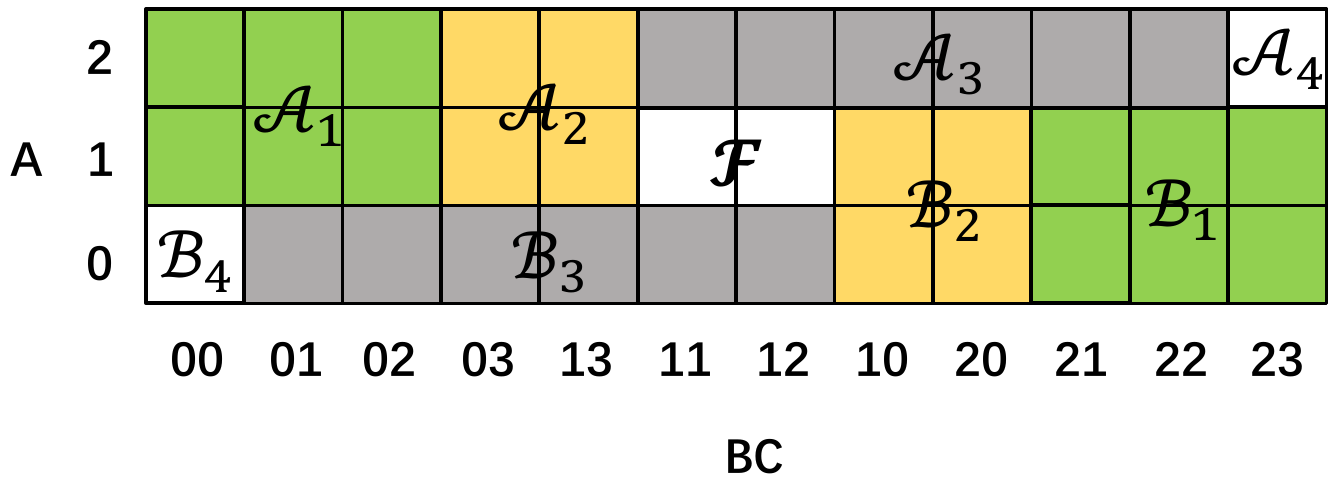}
	\caption{The corresponding $3\times 12$ grid of $\cA_i,\cB_i$ $(i=1,2,3,4)$ and $\cF$ (Eq.~\eqref{OPB334}) in $A|BC$ bipartition.  For example, $\cA_1$ corresponds to the $2\times 3$ grid $\{(1,2)\times (00,01,02)\}$. Moreover,  $\cA_i$ is symmetrical to $\cB_i$ for $1\leq i\leq 4$.  }\label{fig:334}
\end{figure}

\begin{example}\label{example:334}
	In $3\otimes 3\otimes 4$, the  set of states $\cup_{i=1}^3 (\cA_i\cup \cB_i) \cup \cF \cup \{\ket{S}\}$ given by Eq.~\eqref{OPB334} is a UPB of size $28$.
\end{example}
\begin{proof}
	Let $\cH$ be the space spanned by the states in $\cup_{i=1}^3 (\cA_i\cup \cB_i) \cup \cF \cup \{\ket{S}\}$.  For any state $\ket{\psi}\in \cH^{\bot}$, we only need to show that $\ket{\psi}$ must be an entangled state.  We prove it by contradiction. Assume there exists a product state $\ket{\psi}\in \cH^{\bot}$. Let $\cH_1$ be the space spanned by the states in $\cup_{i=1}^3 (\cA_i\cup \cB_i) \cup \cF$. Since $\cH_1\subset \cH$, we have $\cH^{\bot}\subset\cH_1^{\bot}$. Moreover,
	\begin{equation*}
	\cH_1^{\bot}=\text{span}\{\ket{\psi_1(0,0)},\ket{\psi_2(0,0)},\ket{\psi_3(0,0)},\ket{\psi_4},\ket{\phi_1(0,0)},\ket{\phi_2(0,0)},\ket{\phi_3(0,0)},\ket{\phi_4},\ket{\varphi(0)}\}.
	\end{equation*}
	Then $\ket{\psi}$ can be expressed by
	\begin{equation*}
	\ket{\psi}=a_0\ket{\psi_1(0,0)}+b_0\ket{\psi_2(0,0)}+c_0\ket{\psi_3(0,0)}+d_0\ket{\psi_4}+a_1\ket{\phi_1(0,0)}+b_1\ket{\phi_2(0,0)}+c_1\ket{\phi_3(0,0)}+d_1\ket{\phi_4}+e\ket{\varphi(0)},
	\end{equation*}
    where $a_0,b_0,c_0,d_0,a_1,b_1,c_1,d_1,e\in \bbC$. By assumption, $\ket{\psi}$ is a product state, and $\braket{S}{\psi}=0$.

     Next, we consider the matrix form of $\ket{\psi}$ in $A|BC$ bipartition. It corresponds to the $3\times 12$ matrix $M$ in $A|BC$ bipartition, where
    \begin{figure}[H]
    	\centering
    	\includegraphics[scale=0.5]{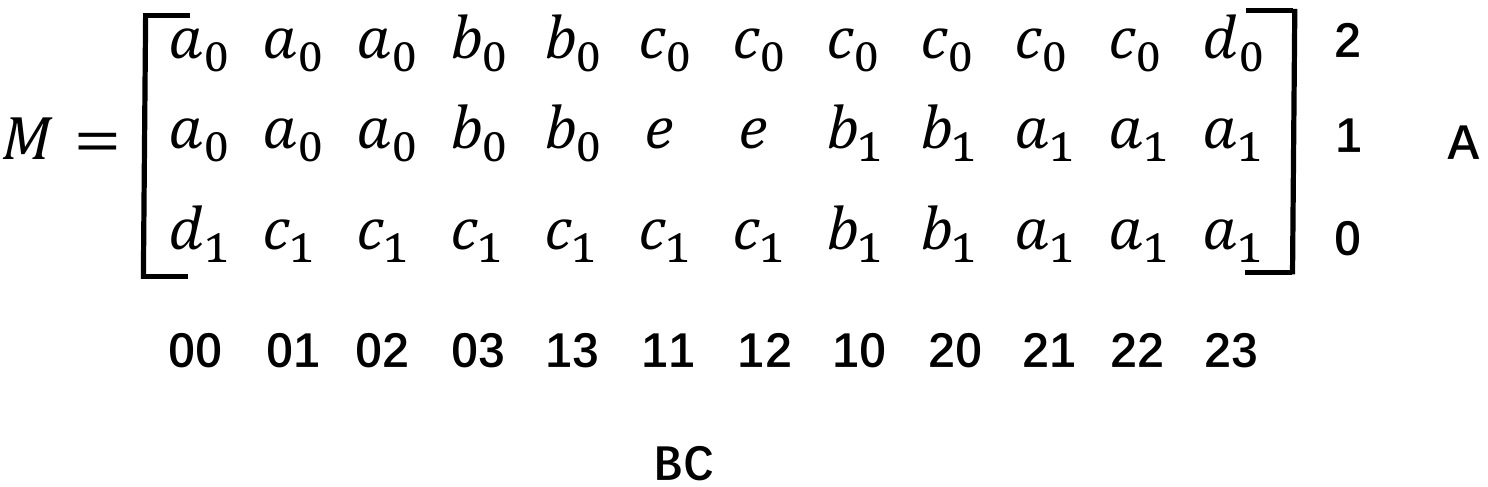}
    \end{figure}
  \noindent and
   \begin{equation}
   \rank(M)=1.
   \end{equation}
   Note that $M$ has a similar structure as Fig.~\ref{fig:334}. It means that $x_0$ is symmetrical to  $x_1$ for $x\in\{a,b,c,d\}$.
   The state $\ket{S}$ corresponds to the matrix $S$ in $A|BC$ bipartition, where
\begin{equation}\label{eq:334_S}
S=\left[\begin{array}{cccccccccccc}
1	&1 &1 &1	&1 &1 &1 &1	 &1 &1 &1 &1\\
1	&1 &1 &1	&1 &1 &1 &1	 &1 &1 &1 &1\\
1	&1 &1 &1	&1 &1 &1 &1	 &1 &1 &1 &1\\
\end{array}\right], \quad \tr(S^{\dagger}M)=\text{sum}(M)=0.
\end{equation}

     Every element of $M$ has coordinate $(A,BC)$. For example, $d_0$ has coordinate $(2,23)$.  If we consider $AB|C$ bipartition,  then we can rearrange the first row of $M$ to a $3\times 4$ matrix, denoted by $M_2$  through $(AB,C)$ coordinates. For example, $d_0$ has coordinate $(22,3)$ in $M_2$. Similarly,  we can rearrange the last row of $M$ to a $3\times 4$ matrix, denoted by $M_0$ through $(AB,C)$ coordinates. Thus we obtain $M_2$ and $M_0$, where
    \begin{figure}[H]
	\centering
	\includegraphics[scale=0.5]{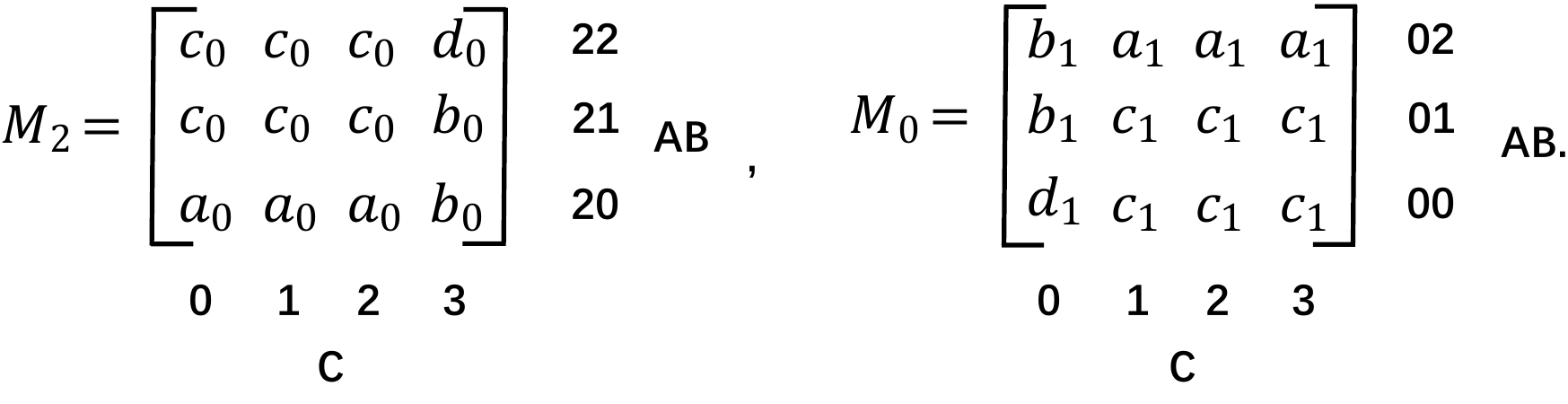}	
\end{figure}

 \noindent
Since  $\ket{\psi}$ is seperable in $AB|C$ bipartition, we must have  \begin{equation}\label{eq:334M_2}
\rank{M_2}=0 \quad \text{or} \quad 1,
\end{equation}
\begin{equation}\label{eq:334M_0}
\rank{M_0}=0 \quad \text{or} \quad 1.
\end{equation}

Assume $a_0\neq 0$. Since $\rank(M)=1$, we have $c_1=d_1$, and $c_0=e=b_1=a_1=d_0$.
\begin{enumerate}[(i)]
	\item If $c_0=0$, then $b_0=0$ by Eq.~\eqref{eq:334M_2}, and $c_1=0$ by Eq.~\eqref{eq:334M_0}, which is impossible by Eq.~\eqref{eq:334_S}.
	\item If $c_0\neq 0$, then $c_0=a_0=b_0$  by Eq.~\eqref{eq:334M_2}, and $c_0=c_1$ by Eq.~\eqref{eq:334M_0}, which is impossible by Eq.~\eqref{eq:334_S}.
\end{enumerate}
Thus we have $a_0=0$. By the symmetry of $M$, we must have $a_1=0$.

Assume $b_0\neq 0$. Since $\rank(M)=1$, we have $c_1=d_1=0$, and $c_0=e=b_1=a_1=d_0=0$. This is impossible by Eq.~\eqref{eq:334_S}. So we must have $b_0=b_1=0$.

Assume $d_1\neq 0$.  Since $\rank(M)=1$, we have $e=c_0=d_0=0$. By Eq.~\eqref{eq:334M_0}, we obtain $c_1=0$.  This is impossible by Eq.~\eqref{eq:334_S}. So we must have $d_0=d_1=0$.

Assume $c_1\neq 0$.  Since $\rank(M)=1$, we have $e=c_0=0$.  This is impossible by Eq.~\eqref{eq:334_S}. So we  have $c_0=c_1=0$.

Since $\text{sum}(M)=0$,  we must have $e=0$, which contradicts $\rank(M)=1$.

Thus $\ket{\psi}$ must be an entangled state, and the set of states $\cup_{i=1}^3 (\cA_i\cup \cB_i) \cup \cF \cup \{\ket{S}\}$ is a UPB.
\end{proof}
\vspace{0.4cm}

Next, we generalize the above UPB to the space $d_A\otimes d_B\otimes d_C$. Let
	\begin{equation}\label{eq:UPB_d_1d_2d_3}
\begin{aligned}
\cA_1&:=\{\ket{\xi_i}_A\ket{0}_B\ket{\eta_k}_C:\ (i,k) \in\bbZ_{d_A-1}\times \bbZ_{d_C-1}\setminus \{(0,0)\}\},\\
\cA_2&:=\{\ket{\xi_i}_A\ket{\eta_j}_B\ket{d_C-1}_C:\ (i,j) \in\bbZ_{d_A-1}\times \bbZ_{d_B-1} \setminus \{(0,0)\}  \},\\			
\cA_3&:=\{\ket{d_A-1}_A\ket{\xi_j}_B\ket{\eta_k}_C:\ (j,k) \in\bbZ_{d_B-1}\times \bbZ_{d_C-1} \setminus \{(0,0)\}\},\\
\cA_4&:=\{|d_A-1\rangle_A|d_B-1\rangle_B|d_C-1\rangle_C \},\\
\cB_1&:=\{\ket{\eta_i}_A\ket{d_B-1}_B\ket{\xi_k}_C:\ (i,k) \in\bbZ_{d_A-1}\times \bbZ_{d_C-1} \setminus \{(0,0)\}\},\\
\cB_2&:=\{\ket{\eta_i}_A\ket{\xi_j}_B\ket{0}_C:\ (i,j) \in\bbZ_{d_A-1}\times \bbZ_{d_B-1} \setminus \{(0,0)\} \},\\
\cB_3&:=\{\ket{0}_A\ket{\eta_j}_B\ket{\xi_k}_C:\ (j,k) \in\bbZ_{d_B-1}\times \bbZ_{d_C-1} \setminus \{(0,0)\}\}, \\
\cB_4&:= \{|0\rangle_A|0\rangle_B|0\rangle_C\},\\
\cF&:=\{\ket{\beta_i}_A\ket{\beta_j}_B\ket{\beta_k}_C:\  (i,j,k)\in\bbZ_{d_A-2}\times\bbZ_{d_B-2}\times \bbZ_{d_C-2}\setminus \{(0,0,0)\}\},\\
\ket{S}&:=\left(\sum_{i=0}^{d_A-1}\ket{i}\right)_A\left(\sum_{j=0}^{d_B-1}\ket{j}\right)_B\left(\sum_{k=0}^{d_C-1}\ket{k}\right)_C,	
\end{aligned}
\end{equation}
where $\ket{\eta_s}_X=\sum_{t=0}^{d_X-2}w_{d_X-1}^{st}\ket{t}_X$, and $\ket{\xi_s}_X=\sum_{t=0}^{d_X-2}w_{d_X-1}^{st}\ket{t+1}_X$ for $s\in \bbZ_{d_X-1 }$, and $X\in\{A,B,C\}$, $\ket{\beta_s}_X=\sum_{t=0}^{d_X-3}w_{d_X-2}^{st}\ket{t+1}_X$ for $s\in \bbZ_{d_X-2}$, and $X\in\{A,B,C\}$.

Note that $\{\ket{\eta_s}_X\}_{s\in\bbZ_{d_X-1}}$, $\{\ket{\xi_s}_X\}_{s\in\bbZ_{d_X-1}}$, and $\{\ket{\beta_s}_X\}_{s\in\bbZ_{d_X-2}}$ are three orthogonal sets, $X\in\{A,B,C\}$,  which are spanned by $\{\ket{t}_X\}_{t=0}^{d_X-2}$,  $\{\ket{t}_X\}_{t=1}^{d_X-1}$, and  $\{\ket{t}_X\}_{t=1}^{d_X-2}$, respectively. This extends the definitions of states in $3\otimes3\otimes4$ in Eq.~\eqref{OPB334}.  The nine subsets $\cA_i,\cB_i$ $(i=1,2,3,4)$ and $\cF$  in $A|BC$ bipartition correspond to the nine blocks of $d_A\times d_Bd_C$ grid in Fig.~\ref{fig:d_Ad_Bd_C}.
Then we have the following proposition.

\begin{figure}[H]
	\centering
	\includegraphics[scale=0.8]{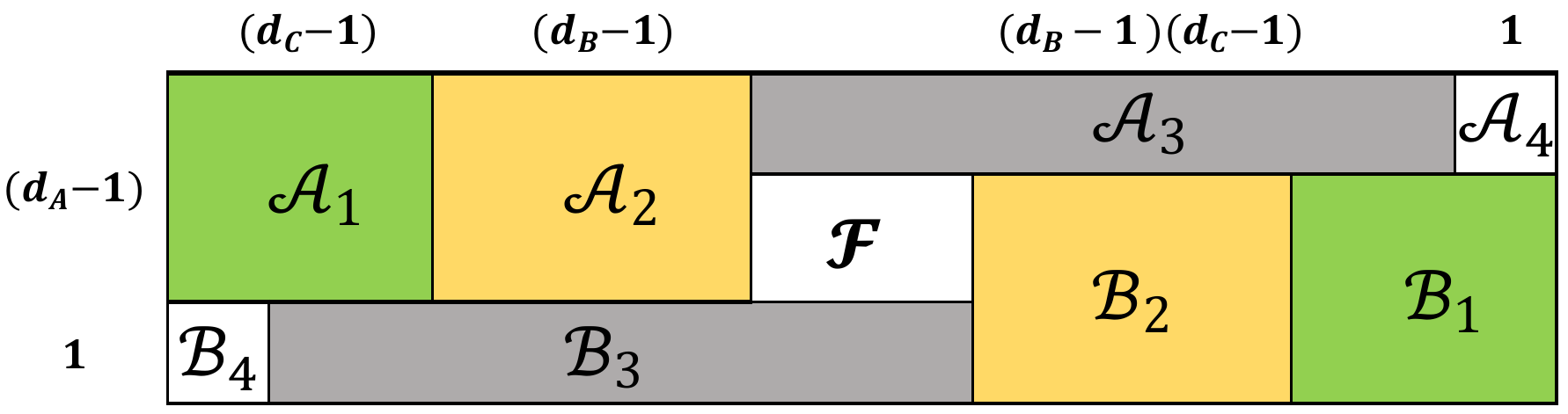}
	\caption{The corresponding $d_A\times d_Bd_C$ grid of $\cA_i,\cB_i$ $(i=1,2,3,4)$ and $\cF$ (Eq.\eqref{eq:UPB_d_1d_2d_3})   in $A|BC$ bipartition.   Moreover,  $\cA_i$ is symmetrical to $\cB_i$ for $1\leq i\leq 4$.  }\label{fig:d_Ad_Bd_C}
\end{figure}

\begin{proposition}\label{pro:3d1d2}
	In $d_A\otimes d_B\otimes d_C$, $3\leq d_A\leq d_B\leq d_C$,  the set of states  $\cup_{i=1}^3(\cA_i\cup\cB_i)\cup \cF \cup \{\ket{S}\}$ given by Eq.~\eqref{eq:UPB_d_1d_2d_3} is a UPB of size $d_Ad_Bd_C-8$.
\end{proposition}

The proof of Proposition~\ref{pro:3d1d2} is given in Appendix~\ref{Appendix:pro3d1d2}. Next, we consider the strong quantum nonlocality for UPBs.
\begin{proposition}\label{pro:dAdBdCofthe}
	In $d_A\otimes d_B\otimes d_C$, $3\leq d_A\leq d_B\leq d_C$,  the UPB $\cup_{i=1}^3(\cA_i\cup\cB_i)\cup \cF \cup \{\ket{S}\}$ given by Eq.~\eqref{eq:UPB_d_1d_2d_3} is strongly nonlocal.	
\end{proposition}
\begin{proof}
 Denote $\cU:=\cup_{i=1}^3(\cA_i\cup\cB_i)\cup \cF \cup \{\ket{S}\}$.  Let $B$ and $C$ come together to  perform a joint  orthogonality-preserving POVM $\{E=M^{\dagger}M\}$, where $E=(a_{ij,k\ell})_{i,k\in\bbZ_{d_B},j,\ell\in \bbZ_{d_C}}$. Then the postmeasurement states $\{\bbI_A\otimes M\ket{\psi}:\ \ket{\psi}\in \cU\}$ should  keep being mutually orthogonal. Assume that $\ket{\psi_1}_A\ket{\psi_2}_B\ket{\psi_3}_C$, $\ket{\varphi_1}_A\ket{\varphi_2}_B\ket{\varphi_3}_C\in  \cU$. Then
	\begin{equation}
	{}_A\bra{\psi_1}{}_B\bra{\psi_2}{}_C\bra{\psi_3}\bbI_A\otimes E\ket{\varphi_1}_A\ket{\varphi_2}_B\ket{\varphi_3}_C=\braket{\psi_1}{\varphi_1}_A({}_B\bra{\psi_2}{}_C\bra{\psi_3} E\ket{\varphi_2}_B\ket{\varphi_3}_C)=0.
	\end{equation}
	If $\braket{\psi_1}{\varphi_1}_A\neq 0$, then ${}_B\bra{\psi_2}{}_C\bra{\psi_3} E\ket{\varphi_2}_B\ket{\varphi_3}_C=0$.
	By using this property, we need to show that $E\propto \bbI$. If we can show that $E\propto \bbI$, then it means that $BC$ can only perform a trivial orthogonality-preserving POVM.  Since the nine subsets $\cA_i,\cB_i$ $(i=1,2,3,4)$, $\cF$  in any bipartition of  $\{A|BC, C|AB, B|CA\}$ correspond to a similar grid as Fig.~\ref{fig:d_Ad_Bd_C}, this implies that any of the party $\{BC, AB, CA\}$ can only perform a trivial orthogonality-preserving POVM. Then we obtain that the UPB $\cU$ is strongly nonlocal. More details for showing $E\propto \bbI$ are given in Appendix~\ref{Appendix:prodAdBdC}.
	\end{proof}
\vspace{0.4cm}

Note that the states in $\cA_i$ or $\cB_i$ $(i=1,2,3,4)$ in Eq.~\eqref{eq:UPB_d_1d_2d_3} are defined by the outermost layer of a $d_A\times d_B\times d_C$ cube, and the states in $\cF$ are just  defined by all inner cells. By observing this, we can construct more strongly nonlocal UPBs in  $d_A\otimes d_B\otimes d_C$  by continuing to decompose $\cF$ in Fig.~\ref{fig:d_Ad_Bd_C} by the similar tiling method. Suppose we are on the $n$-th layer from outside to inside, $0\leq n\leq \fl{d_A-3}{2}$. Let $X_n:=d_X-2n$ for $X\in \{A,B,C\}$. Then we can define the following states,
 \begin{equation}\label{eq:outlayer}
 \begin{aligned}
 \cA_1^{(n)}&:=\{\ket{\xi_i^{(n)}}_A\ket{n}_B\ket{\eta_k^{(n)}}_C:\ (i,k) \in\bbZ_{A_n-1}\times \bbZ_{C_n-1}\setminus \{(0,0)\}\},\\
 \cA_2^{(n)}&:=\{\ket{\xi_i^{(n)}}_A\ket{\eta_j^{(n)}}_B\ket{d_C-1-n}_C:\ (i,j) \in\bbZ_{A_n-1}\times \bbZ_{B_n-1} \setminus \{(0,0)\} \},\\			
 \cA_3^{(n)}&:=\{\ket{d_A-1-n}_A\ket{\xi_j^{(n)}}_B\ket{\eta_k^{(n)}}_C:\ (j,k) \in\bbZ_{B_n-1}\times \bbZ_{C_n-1}\setminus \{(0,0)\}\},\\
 \cA_4^{(n)}&:=\{|d_A-1-n\rangle_A|d_B-1-n\rangle_B|d_C-1-n\rangle_C\},\\
 \cB_1^{(n)}&:=\{\ket{\eta_i^{(n)}}_A\ket{d_B-1-n}_B\ket{\xi_k^{(n)}}_C:\ (i,k) \in\bbZ_{A_n-1}\times \bbZ_{C_n-1}\setminus \{(0,0)\} \},\\
 \cB_2^{(n)}&:=\{\ket{\eta_i^{(n)}}_A\ket{\xi_j^{(n)}}_B\ket{n}_C:\ (i,j)\in\bbZ_{A_n-1}\times \bbZ_{B_n-1}\setminus \{(0,0)\} \},\\
 \cB_3^{(n)}&:=\{\ket{n}_A\ket{\eta_j^{(n)}}_B\ket{\xi_k^{(n)}}_C:\ (j,k) \in\bbZ_{B_n-1}\times \bbZ_{C_n-1}\setminus \{(0,0)\}\}, \\
 \cB_4^{(n)}&:= \{|n\rangle_A|n\rangle_B|n\rangle_C\},\\
 \cF^{(n)}&:=\{\ket{\beta_i^{(n)}}_A\ket{\beta_j^{(n)}}_B\ket{\beta_k^{(n)}}_C:\ (i,j,k)\in\bbZ_{A_n-2}\times\bbZ_{B_n-2}\times \bbZ_{C_n-2}\setminus \{(0,0,0)\}\},\\
 \ket{S}&=\left(\sum_{i=0}^{d_A-1}\ket{i}\right)_A\left(\sum_{j=0}^{d_B-1}\ket{j}\right)_B\left(\sum_{k=0}^{d_C-1}\ket{k}\right)_C,		
 \end{aligned}
 \end{equation}
  where $\ket{\eta_s^{(n)}}_X=\sum_{t=n}^{X_n+n-2}w_{X_n-1}^{s(t-n)}\ket{t}_X$, and $\ket{\xi_s^{(n)}}_X=\sum_{t=n}^{X_n+n-2}w_{X_n-1}^{s(t-n)}\ket{t+1}_X$, for $s\in\bbZ_{X_n-1}$, and  $X\in\{A,B,C\}$,  $\ket{\beta_s^{(n)}}_X=\sum_{t=n}^{X_n+n-3}w_{X_n-2}^{s(t-n)}\ket{t+1}_X$ for $s\in \bbZ_{X_n-2}$, and $X\in\{A,B,C\}$.

Note that $\{\ket{\eta_s^{(n)}}_X\}_{s\in\bbZ_{X_n-1}}$, $\{\ket{\xi_s^{(n)}}_X\}_{s\in\bbZ_{X_n-1}}$, and $\{\ket{\beta_s^{(n)}}_X\}_{s\in\bbZ_{X_n-2}}$ are three  orthogonal sets, $X\in\{A,B,C\}$, which are spanned by $\{\ket{t}_X\}_{t=n}^{X_n+n-2}$,  $\{\ket{t}_X\}_{t=n+1}^{X_n+n-1}$, and  $\{\ket{t}_X\}_{t=n+1}^{X_n+n-2}$, respectively. This extends the definition of states in Eq.~\eqref{eq:UPB_d_1d_2d_3} from $n=0$ to general $n$. Specially, $\cA_i^{(0)}, \cB_i^{(0)}$ $(i=1,2,3,4)$, $\cF^{(0)}$ are $\cA_i, \cB_i$ $(i=1,2,3,4)$, $\cF$ of Eq.~\eqref{eq:UPB_d_1d_2d_3} exactly, which correspond to Fig.~\ref{fig:d_Ad_Bd_C} in $A|BC$ bipartition. Next, when $d_A\geq 5$, $\cA_i^{(0)}, \cB_i^{(0)}$, $\cA_i^{(1)}, \cB_i^{(1)}$ $(i=1,2,3,4)$, $\cF^{(1)}$ correspond to Fig.~\ref{fig:d_Ad_Bd_C_A_BC} in $A|BC$ bipartition. Then we have the following theorem.

 \begin{figure}[H]
 	\centering
 	\includegraphics[scale=0.6]{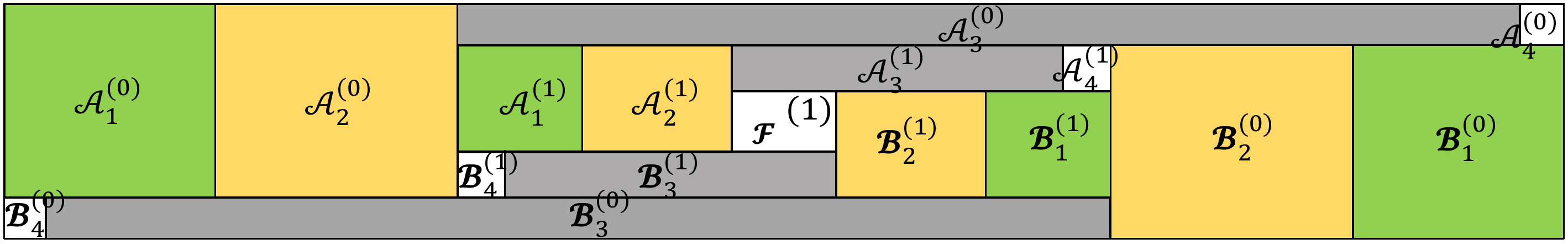}
 	\caption{The corresponding $d_A\times d_Bd_C$ grid of $\cA_i^{(0)}, \cB_i^{(0)}$, $\cA_i^{(1)}, \cB_i^{(1)}$ $(i=1,2,3,4)$, $\cF^{(1)}$ (Eq.~\eqref{eq:outlayer}) in $A|BC$ bipartition,  where $d_A\geq 5$.  }\label{fig:d_Ad_Bd_C_A_BC}
 \end{figure}

\begin{theorem}\label{thm:d1d2d3}
	In $d_A\otimes d_B\otimes d_C$, $3\leq d_A\leq d_B\leq d_C$, for any $0\leq n\leq \fl{d_A-3}{2}$, the set
	
	\begin{equation*}
	\cU_n:=\cup_{t=0}^n(\cup_{i=1}^3(\cA_i^{(t)}\cup\cB_i^{(t)}))\cup \cF^{(n)}\cup \{\ket{S}\}
	\end{equation*}  given by Eq.~\eqref{eq:outlayer} is a strongly nonlocal UPB of size $d_Ad_Bd_C-8(n+1)$.	
\end{theorem}

The proof of Theorem~\ref{thm:d1d2d3} is given in Appendix~\ref{Appendix:thmd1d2d3}. By now, we have shown strongly nonlocal UPBs exist in any three-partite systems. For four-partite systems,  strongly nonlocal UPBs are shown to exist in Supplementary material \cite{Supplementary} by using the similar method. In \cite{shi20211}, the authors gave a decomposition of the $5$-dimensional hypercube, which may be used for constructing strongly nonlocal UPBs in any five-partite systems. However, it  requires more calculations. We leave this as an open question.

\section{Entanglement-assisted discrimination}\label{sec:discrim}

In this section, we investigate the local discrimination of the strongly nonlocal UPB in $3\otimes 3\otimes 4$ in Example~\ref{example:334} by using entanglement as a resource. In a protocol of local quantum state discrimination, a multipatite quantum system is prepared with
a state which is secretly chosen from a known set, and the
purpose is to determine in which state the system is by using LOCC only. If the set of states is locally indistinguishable, then additional entanglement resources can assist for perfect discrimination. This process is called entanglement-assisted discrimination \cite{cohen2008understanding}.

\begin{figure}[H]
	\centering
	\includegraphics[scale=1]{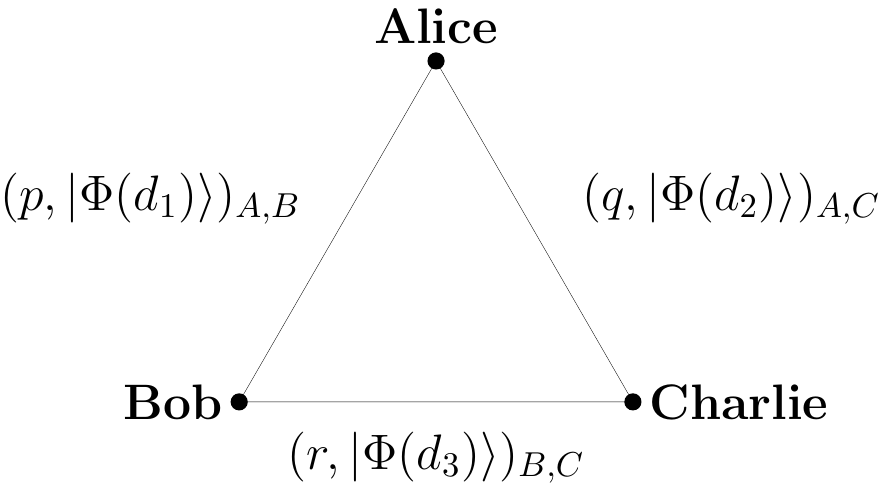}
\caption{Entanglement-assisted discrimination. A strongly nonlocal tripartite UPB is shared by Alice, Bob and Charlie. Their aim is to distinguish these states in the UPB by using LOCC only. For perfect discrimination of the UPB, an amount $p$ of the maximally entangled state $\ket{\Phi(d_1)}_{A,B}$ is shared by Alice and Bob; an amount $q$ of the maximally entangled state $\ket{\Phi(d_2)}_{A,C}$ is shared by Alice and Charlie, and an amount $r$ of the maximally entangled state $\ket{\Phi(d_2)}_{B,C}$ is shared by Bob and Charlie, where $\ket{\Phi(d)}_{X,Y}=\sum_{k=0}^{d-1}\ket{k}_X\ket{k}_Y$ for $X,Y\in\{A,B,C\}$.  It consumes $p\log_2d_1+q\log_2d_2+r\log_2d_3$ ebits of entanglement resource.     }\label{Fig:distinguish}
	\end{figure}

Assume a strongly nonlocal tripartite UPB is shared by Alice, Bobs and Charlie. Since a strongly nonlocal UPB is locally indistinguishable in any bipartition, a perfect discrimination of this set needs a resource state that must be entangled in all bipartitions. The configuration of entanglement resources can be described by $\{(p,\ket{\Phi(d_1)})_{A,B}, (q,\ket{\Phi(d_2)})_{A,C},(r,\ket{\Phi(d_3)})_{B,C}\}$ \cite{Sumit2019Genuinely,2020Strong}. Here, $(p,\ket{\Phi(d_1)})_{A,B}$ means that an amount $p$ of the maximally entangled state $\ket{\Phi(d_1)}$  is shared by Alice and Bob; $(q,\ket{\Phi(d_2)})_{A,C}$ means that an amount $q$ of the maximally entangled state $\ket{\Phi(d_2)}$  is shared by Alice and Charlie; $(r,\ket{\Phi(d_3)})_{B,C}$ means that an amount $r$ of the maximally entangled state $\ket{\Phi(d_3)}$  is shared by Bob and Charlie, where $\ket{\Phi(d)}_{X,Y}=\sum_{k=0}^{d-1}\ket{k}_X\ket{k}_Y$ for $X,Y\in\{A,B,C\}$. It consumes $p\log_2d_1+q\log_2d_2+r\log_2d_3$ ebits of entanglement resource. See Fig.~\ref{Fig:distinguish} for this configuration of entanglement-assisted discrimination.  Now, we give a  discrimination protocol  for the strongly nonlocal UPB given by Eq.~\eqref{OPB334}.

\begin{proposition}\label{pro:dis}
	In $3\otimes 3\otimes 4$, the strongly nonlocal UPB $\cup_{i=1}^3(\cA_i\cup\cB_i)\cup\cF\cup \{\ket{S}\}$ given by Eq.~\eqref{OPB334} can be locally distinguished by using $\{(1,\ket{\Phi(2)})_{A,B}, (0,\ket{\Phi(d_2)})_{A,C},(1,\ket{\Phi(2)})_{B,C}\}$ for any positive $d_2$, which consumes $2$ ebits of entanglement resource.
\end{proposition}

The proof of Proposition~\ref{pro:dis} is given in Appendix~\ref{Appendix:prodis}.
 If we only use the teleportation-based protocol \cite{Bennett1993Teleporting,Sumit2019Genuinely}, the strongly nonlocal UPB given by Eq.~\eqref{OPB334} can be locally distinguished by using  $\{(1,\ket{\Phi(3)})_{A,B}, (0,\ket{\Phi(d_2)})_{A,C},(1,\ket{\Phi(3)})_{B,C}\}$, which consumes $2\log_23$ ebits of entanglement resource. Thus the protocol in Proposition~\ref{pro:dis}  consumes less entanglement resource than the teleportation-based protocol.

\section{Conclusion and discussion}\label{sec:con}
In this work, we have constructed a series of strongly nonlocal UPBs of different sizes in any possible three and four-partite systems. We have also proposed an entanglement-assisted protocol for local discrimination of the strongly nonlocal UPB in $3\otimes 3\otimes 4$.  All of our  UPBs of large sizes can be used  to construct PPT entangled states with small rank. For example, in Proposition~\ref{pro:3d1d2}, for a normalized UPB $\{\ket{\psi_i}\}_{i=1}^{d_Ad_Bd_C-8}$ in $d_A\otimes d_B\otimes d_C$  for $3\leq d_A\leq d_B\leq d_C$, the mixed state
\begin{equation*}
\rho=\frac{1}{8}(\bbI-\sum_{i=1}^{d_Ad_Bd_C-8}\ketbra{\psi_i}{\psi_i})
\end{equation*}
is a PPT entangled state with rank $8$. Further our UPBs of large sizes can be also used for low-rank noisy bound entangled states which satisfy the range criterion \cite{Bej}. Bipartite noisy bound entangled states that satisfy the range criterion were shown in \cite{halder2019construction}. There are some open questions left. Whether  UPBs with the minimum size can show strong quantum nonlocality? Whether one can generalize our constructions on $d_1\otimes d_2 \otimes \cdots \otimes d_n$ for $n\geq 5$?

\section*{Acknowledgments}
\label{sec:ack}	
The authors are very grateful to the reviewers for providing
many useful suggestions which have greatly improved the presentation of our paper.  F.S. and X.Z. were supported by the NSFC under Grants No. 11771419 and No. 12171452, the Anhui Initiative in Quantum Information Technologies under Grant No. AHY150200, and the National Key Research and Development Program of China 2020YFA0713100. M.-S.L. was supported by
NSFC (Grants No. 12005092) and the China Postdoctoral Science Foundation (2020M681996). L.C. was supported by the  NNSF of China (Grant No. 11871089), and the Fundamental Research Funds for the Central Universities (Grant No. ZG216S2005).

\appendix

\section{The proof of Proposition~\ref{pro:3d1d2}}\label{Appendix:pro3d1d2}
\begin{proof}
	 For the same discussion as Example~\ref{example:334}, we can assume that $\ket{\psi}$ is a product state in the complementary space of the space spanned by the states in $\cup_{i=1}^3(\cA_i\cup\cB_i)\cup \cF \cup \{\ket{S}\}$. Then we consider the matrix form of  $\ket{\psi}$ in $A|BC$ and $AB|C$ bipartitions. First, we  consider $A|BC$ bipartition. We have
	\begin{equation}\label{eq:M}
	M=\left[\begin{array}{ccccccccccccccccc}
	a_0 &a_0 &\cdots &a_0 &b_0 &\cdots &b_0 &c_0 &\cdots &c_0 &c_0 &\cdots &c_0 &c_0 &\cdots &c_0 &d_0\\
	a_0 &a_0 &\cdots &a_0 &b_0 &\cdots &b_0 &e &\cdots &e &b_1 &\cdots &b_1 &a_1&\cdots &a_1 &a_1\\
  \vdots &\vdots &\ddots  &\vdots&\vdots&\ddots&\vdots&\vdots&\ddots&\vdots&\vdots&\ddots&\vdots&\vdots&\ddots&\vdots&\vdots\\
	a_0 &a_0 &\cdots &a_0 &b_0 &\cdots &b_0 &e &\cdots &e &b_1 &\cdots &b_1 &a_1&\cdots &a_1 &a_1\\
	d_1 &c_1 &\cdots &c_1 &c_1 &\cdots &c_1 &c_1 &\cdots &c_1 &b_1 &\cdots &b_1 &a_1 &\cdots &a_1 &a_1
	\end{array}	\right],
		\end{equation}
and
	\begin{equation}\label{eq:rank3d_1d_2}
     \rank(M)=1, \quad	\text{sum}(M)=0.
	\end{equation}
Then we consider $AB|C$ bipartition. We can rearrange  the first row of $M$ to the $d_B\times d_C$ matrix $M_{(d_A-1)}$  through $(AB,C)$ coordinates of $M$, and rearrange the last row of $M$ to the $d_B\times d_C$ matrix $M_{0}$  through $(AB,C)$ coordinates of $M$, where
	\begin{equation}
	M_{(d_A-1)}=\begin{bmatrix}\label{eq:M_23d_1d_2}
	c_0 &c_0 &\cdots &c_0 &d_0\\
	c_0 &c_0 &\cdots &c_0 &b_0\\
	\vdots 	&\vdots  	&\ddots  	&\vdots 	&\vdots\\
	c_0 &c_0 &\cdots &c_0 &b_0\\
	a_0 &a_0 &\cdots &a_0 &b_0\\
	\end{bmatrix}, \quad \rank(M_{(d_A-1)})=0\ \ \text{or} \ \  1,
	\end{equation}
and
		\begin{equation}
	M_0=\begin{bmatrix}\label{eq:M_03d_1d_2}
	b_1 &a_1 &\cdots &a_1 &a_1\\
	b_1 &c_1 &\cdots &c_1 &c_1\\
	\vdots 	&\vdots  	&\ddots  	&\vdots 	&\vdots\\
	b_1 &c_1 &\cdots &c_1 &c_1\\
	d_1 &c_1 &\cdots &c_1 &c_1\\
	\end{bmatrix}, \quad \rank(M_0)=0\ \ \text{or} \ \  1.
	\end{equation}
For the same proof as Example~\ref{example:334}, we can show that $M$ do not exist by Eqs.~\eqref{eq:rank3d_1d_2}, \eqref{eq:M_23d_1d_2} and \eqref{eq:M_03d_1d_2}. Thus  $\ket{\psi}$ must be an entangled state, and $\cup_{i=1}^3(\cA_i\cup\cB_i)\cup \cF \cup \{\ket{S}\}$ is a UPB.
\end{proof}
\vspace{0.4cm}

\section{The proof of Proposition~\ref{pro:dAdBdCofthe}}\label{Appendix:prodAdBdC}
\begin{proof}
	First of all, we need to introduce some notations which have been introduced by \cite{shi2021}. Let $\mathcal{S}=\{\ket{\psi_1}_A\ket{\psi_2}_B\ket{\psi_3}\}$ be a tripartite orthogonal product set. Define
    \begin{equation*}
    \mathcal{S}(|\psi\rangle_A):=\{ |\psi_2\rangle_B|\psi_3\rangle_C   :\   |\psi\rangle_A |\psi_2\rangle_B|\psi_3\rangle_C \in  \mathcal{S}\}.
   \end{equation*}
   Moreover, define  $\mathcal{S}^{(A)}=\{\ket{j}_B\ket{k}_C:\ j,k\in\bbZ_n\}$ as the support of $\mathcal{S}(|\psi\rangle_A)$  which spans  $\mathcal{S}(|\psi\rangle_A)$.
    For example, in Eq.~\eqref{OPB334}, $\cA_1:=\{\ket{\xi_i}_A\ket{0}_B\ket{\eta_j}_C:\ (i,j)\neq(0,0) \in\bbZ_2\times \bbZ_3\}$. Then $\cA_1(\ket{\xi_1}_A)=\{\ket{0}_A\ket{\eta_j}_C\}_{j\in \bbZ_3}$, $\cA_1^{(A)}=\{\ket{0}_B\ket{0}_C,\ket{0}_B\ket{1}_C,\ket{0}_B\ket{2}_C\}$, and $\cA_1(\ket{\xi_1}_A)$ is spanned by $\cA_1^{(A)}$. Actually,  $\cA_i^{(A)},\cB_i^{(A)}$ $(i=1,2,3,4)$, $\cF^{(A)}$ in Eq.~\eqref{OPB334}    can be easily observed by Fig.~\ref{fig:334}. They are the projection sets of $\cA_i,\cB_i$ $(i=1,2,3,4)$, $\cF$ in $BC$ party in Fig.~\ref{fig:334}. Now, we give three steps for the proof.
	
	\noindent{\bf Step 1} Since $\braket{\xi_1}{\eta_1}_A\neq 0$, applying Block Zeros Lemma to any two elements of $\{\cA_{1}(\ket{\xi_1}_A)$, $\cA_{2}(\ket{\xi_1}_A)$, $\cB_{2}(\ket{\eta_1}_A)$, $ \cB_{1}(\ket{\eta_1}_A)\}$, we obtain
	\begin{equation}\label{eq:four_orthogonal}
   {}_{\cA_{i}^{(A)}}E_{\cA_{j}^{(A)}}=\mathbf{0}, \ {}_{\cA_{i}^{(A)}}E_{\cB_{k}^{(A)}}=\mathbf{0}, \ {}_{\cB_{k}^{(A)}}E_{\cB_{\ell}^{(A)}}=\mathbf{0}, \ {}_{\cB_{k}^{(A)}}E_{\cA_{i}^{(A)}}=\mathbf{0},
 \end{equation}
for $1\leq i\neq j\leq 2$, $1\leq k\neq \ell\leq 2$.
	Note that if $d_A=3$, $\ket{\beta_i}_A$ must be  $\ket{\beta_0}_A$   in $\cF$. So we can not apply Block Zeros Lemma to $\cF(\ket{\beta_1}_A)$ and $\cA_{1}(\ket{\xi_1}_A)$. In order to obtain ${}_{\cF^{(A)}}E_{\cA_{1}^{(A)}}=\mathbf{0}$,   we consider $\cA_{3}(\ket{d_A-1}_A)$ and $\cA_{1}(\ket{\xi_1}_A)$. Then for  $(j, k)\in\bbZ_{d_B-1}\times \bbZ_{d_C-1}\setminus\{(0,0)\}$, and $i\in\bbZ_{d_C-1}$, we have
	\begin{equation}\label{eq:center_1}
		{}_B\bra{\xi_j}{}_C\bra{\eta_k}E\ket{0}_B\ket{\eta_i}_C=		{}_B\left(\sum_{t_1=0}^{d_B-2}w_{d_B-1}^{-jt_1}\bra{t_1+1}\right){}_C\left(\sum_{t_2=0}^{d_C-2}w_{d_C-1}^{-kt_2}\bra{t_2}\right)E\ket{0}_B\left(\sum_{t_3=0}^{d_C-2}w_{d_C-1}^{it_3}\ket{t_3}\right)_C=0.
	\end{equation}
 We have shown that ${}_B\bra{j+1}{}_C\bra{0}E\ket{0}_B\ket{i}_C=0$ for $j\in\bbZ_{d_B-1}$, $i\in\bbZ_{d_C-1}$, and ${}_B\bra{d_B-1}{}_C\bra{k+1}E\ket{0}_B\ket{i}_C=0$ for $k,i\in\bbZ_{d_C-1}$ by Eq.~\eqref{eq:four_orthogonal}. Then Eq.~\eqref{eq:center_1} can be expressed by
	\begin{equation*}
		{}_B\left(\sum_{t_1=0}^{d_B-3}w_{d_B-1}^{-jt_1}\bra{t_1+1}\right){}_C\left(\sum_{t_2=1}^{d_C-2}w_{d_C-1}^{-kt_2}\bra{t_2}\right)E\ket{0}_B\left(\sum_{t_3=0}^{d_C-2}w_{d_C-1}^{it_3}\ket{t_3}\right)_C=0,
\end{equation*}
i.e.
	\begin{equation*}
\sum_{t_1=0}^{d_B-3}\sum_{t_2=1}^{d_C-2}\sum_{t_3=0}^{d_C-2}w_{d_B-1}^{-jt_1}w_{d_C-1}^{-kt_2}w_{d_C-1}^{it_3}{}_B\bra{t_1+1}{}_C\bra{t_2}E\ket{0}_B\ket{t_3}_C=0,
\end{equation*}
for any $0\leq j\leq d_B-3$,  $1\leq k\leq d_C-2$, $0\leq i\leq d_C-2$. It means that
\begin{equation*}
[H_1^{\dagger}\otimes H_2^{\dagger}\otimes H_3]X=\textbf{0},
\end{equation*}
where
\begin{equation*}
H_1=\begin{pmatrix}
1 &1 &\cdots &1\\
1 &w_{d_B-1} &\cdots &w_{d_B-1}^{(d_B-3)}\\
\vdots &\vdots &\ddots &\vdots\\
1 &w_{d_B-1}^{(d_B-3)} &\cdots &w_{d_B-1}^{(d_B-3)^2}
\end{pmatrix}, \quad
H_2=\begin{pmatrix}
w_{d_C-1} &w_{d_C-1}^2 &\cdots &w_{d_C-1}^{(d_C-2)}\\
w_{d_C-1}^2 &w_{d_C-1}^4  &\cdots &w_{d_C-1}^{2(d_C-2)}\\
\vdots &\vdots &\ddots &\vdots\\
w_{d_C-1}^{(d_C-2)} &w_{d_C-1}^{2(d_C-2)}&\cdots &w_{d_C-1}^{(d_C-2)^2}
\end{pmatrix},
\end{equation*}
$H_3=(w_{d_C-1}^{ij})_{i,j\in\bbZ_{d_C-1}}$
and $X$ is a column vector,
\begin{equation*}
X=({}_B\bra{t_1+1}{}_C\bra{t_2}E\ket{0}_B\ket{t_3}_C)_{\{0\leq t_1\leq d_B-3, \ 1\leq t_2\leq d_C-2, \ 0\leq t_3\leq d_C-2\}}.
\end{equation*}
Since $H_1,H_2,H_3$ are all full-rank matrices, it implies that $H_1^{\dagger}\otimes H_2^{\dagger}\otimes H_3$ is a full-rank matrix. Then $X=\textbf{0}$, i.e.
\begin{equation*}
{}_B\bra{t_1+1}{}_C\bra{t_2}E\ket{0}_B\ket{t_3}_C=0, \quad \text{for} \ 0\leq t_1\leq d_B-3, \ 1\leq t_2\leq d_C-2, \ 0\leq t_3\leq d_C-2.
\end{equation*}
It also means that
\begin{equation}\label{eq:FA1}
{}_{\cF^{(A)}}E_{\cA_{1}^{(A)}}=\mathbf{0}.
\end{equation}
By using $\cA_{3}(\ket{d_A-1}_A)$ and $\cA_{2}(\ket{\xi_1}_A)$, we can also show that
\begin{equation}\label{eq:FA2}
{}_{\cF^{(A)}}E_{\cA_{2}^{(A)}}=\mathbf{0}
\end{equation}
by the similar discussion as above.
Further, by the symmetry of Fig.~\ref{fig:d_Ad_Bd_C}, we can also obtain that
\begin{equation}\label{eq:FB1B2}
{}_{\cF^{(A)}}E_{\cB_{1}^{(A)}}=\mathbf{0}, \quad {}_{\cF^{(A)}}E_{\cB_{2}^{(A)}}=\mathbf{0}.
\end{equation}
Thus, by Eqs.~\eqref{eq:four_orthogonal}, \eqref{eq:FA1}, \eqref{eq:FA2} and \eqref{eq:FB1B2}, $E$ is a block diagonal matrix. It can be expressed by
\begin{equation}\label{eq:matrixE}
E=E_{\cA_1^{(A)}}\oplus E_{\cA_2^{(A)}}\oplus E_{\cF^{(A)}}\oplus E_{\cB_2^{(A)}}\oplus E_{\cB_1^{(A)}}.
\end{equation}

\noindent{\bf Step 2}  Considering $\ket{S}$ and $\{\ket{\beta_0}_A\ket{\beta_j}_B\ket{\beta_k}_C \}_{(j,k) \in\bbZ_{d_B-2}\times\bbZ_{d_C-2}\setminus\{(0,0)\}}\subset \cF$, where $d_C\geq 4$, then by Eq.~\eqref{eq:matrixE}, we have
\begin{equation*}
{}_B\left(\sum_{i_1=0}^{d_B-1}\bra{i_1}\right){}_C\left(\sum_{i_2=0}^{d_C-1}\bra{i_2}\right)E\ket{\beta_j}_B\ket{\beta_k}_C={}_B\left(\sum_{i_1=1}^{d_B-2}\bra{i_1}\right){}_C\left(\sum_{i_2=1}^{d_C-2}\bra{i_2}\right)E\ket{\beta_j}_B\ket{\beta_k}_C=0,
\end{equation*}
for  $(j,k)\in\bbZ_{d_B-2}\times \bbZ_{d_C-2}\setminus\{(0,0)\}$.
Moreover, we have
\begin{equation*}
\left(\sum_{i_1=1}^{d_B-2}\ket{i_1}\right)_B\left(\sum_{i_2=1}^{d_C-2}\ket{i_2}\right)_C=\ket{\beta_0}_B\ket{\beta_0}_C.
\end{equation*}
Therefore, by using the states in $\{\ket{S}\}\cup\{\ket{\beta_0}_A\ket{\beta_j}_B\ket{\beta_k}_C\}_{(j,k) \in\bbZ_{d_B-2}\times\bbZ_{d_C-2}}\setminus\{(0,0)\}$, we obtain
\begin{equation}\label{eq:exist}
{}_B\bra{\beta_i}{}_C\bra{\beta_j}E\ket{\beta_k}_B\ket{\beta_\ell}_C=0, \quad \text{for} \  (i, j)\neq(k,\ell)\in\bbZ_{d_B-2}\times\bbZ_{d_C-2}.
\end{equation}
By Eq.~\eqref{eq:exist}, there exists a real number $e_{s,t}$  ($E^{\dagger}=E$)  for $(s,t)\in\bbZ_{d_B-2}\times\bbZ_{d_C-2}$,  such that
\begin{equation}\label{eq:EF}
E_{\cF^{(A)}}=\sum_{s=0}^{d_B-3}\sum_{t=0}^{d_C-3}e_{s,t}\ket{\beta_s}_B\bra{\beta_s}\otimes\ket{\beta_t}_C\bra{\beta_t}.
\end{equation}
Note that  Eq.~\eqref{eq:EF} also holds for $d_C=3$ (in this case, $E_{\cF^{(A)}}=e_{0,0}\ket{\beta_0}_B\bra{\beta_0}\otimes\ket{\beta_0}_C\bra{\beta_0}$, and $\ket{\beta_0}_X=\ket{1}_X$ for $X\in\{B,C\}$).
Next, by using the states in $\cA_1(\ket{\xi_1}_A)$, we have
\begin{equation*}
{}_B\bra{0}{}_C\bra{\eta_i}E\ket{0}_B\ket{\eta_j}_C=0, \quad \text{for} \  i\neq j\in\bbZ_{d_C-1}.
\end{equation*}
Then there exists a real number $a_s$ for $s\in\bbZ_{d_C-1}$  such that
\begin{equation*}
E_{\cA_1^{(A)}}=\sum_{s=0}^{d_C-2}a_s\ket{0}_B\bra{0}\otimes\ket{\eta_s}_C\bra{\eta_s}.
\end{equation*}
In the same way, there exist real numbers $a_s,b_s,c_t,d_t,e_{s,t}$ such that the operator
\begin{equation}\label{eq:operatorM}
\begin{aligned}
E=&\sum_{s=0}^{d_C-2}a_s\ket{0}_B\bra{0}\otimes\ket{\eta_{s}}_C\bra{\eta_{s}}+\sum_{s=0}^{d_B-2} b_s\ket{\eta_{s}}_B\bra{\eta_{s}}\otimes\ket{d_C-1}_C\bra{d_C-1}+\sum_{s=0}^{d_B-3}\sum_{t=0}^{d_C-3} e_{s,t}\ket{\beta_s}_B\bra{\beta_s}\otimes\ket{\beta_t}_C\bra{\beta_t}\\ &+\sum_{t=0}^{d_B-2}c_t\ket{\xi_{t}}_B\bra{\xi_{t}}\otimes\ket{0}_C\bra{0}+\sum_{t=0}^{d_C-2} d_t\ket{d_B-1}_B\bra{d_B-1}\otimes\ket{\xi_{t}}_C\bra{\xi_{t}}.
\end{aligned}
\end{equation}	
By using 	those states $ \{\ket{0}_A\ket{\eta_i}_B\ket{\xi_j}_C\}_{(i,j)\in\mathbb{Z}_{d_B-1}\times \mathbb{Z}_{d_C-1}\setminus\{(0,0)\}}=\cB_3$, we can show that
\begin{equation*}
{}_B\bra{\eta_k} {}_C\bra{\xi_\ell} E \ket{\eta_i}_B\ket{\xi_j}_C=0,  \quad \text{for} \ (k,\ell)\neq (i,j) \in\mathbb{Z}_{d_B-1}\times \mathbb{Z}_{d_C-1}\setminus \{(0,0)\}.
\end{equation*}
Assume $k=0,\ell\neq 0, i\neq 0, j=0$. By Eq.~\eqref{eq:operatorM}, we have
\begin{equation}\label{eq:thezero}
\begin{aligned}
0=&{}_B\bra{\eta_0} {}_C\bra{\xi_\ell}  E \ket{\eta_i}_B\ket{\xi_0}_C\\
=&\sum_{s=0}^{d_C-2}a_s   \braket{ \eta_0}{0}_B \braket{0}{\eta_i}_B  \braket{\xi_\ell}{\eta_s}_C \braket{\eta_s}{\xi_0}_C +
\sum_{s=0}^{d_B-3}\sum_{t=0}^{d_C-3}e_{s,t}\braket{\eta_0}{\beta_s}_B\braket{\beta_s}{\eta_i}_B\braket{\xi_\ell}{\beta_t}_C\braket{\beta_t}{\xi_0}_C\\
=&\sum_{s=0}^{d_C-2}a_s\braket{\xi_\ell}{\eta_s}_C \braket{\eta_s}{\xi_0}_C+w_{d_C-1}^{\ell}(d_B-2)(d_C-2)e_{0,0}.
\end{aligned}
\end{equation}
There are three cases for the terms in the summation of the last equality.
\begin{enumerate}[(a)]
	\item If $s=0$, then
	\begin{equation*}
	\braket{\xi_\ell}{\eta_0}_C\braket{\eta_0}{\xi_0}_C=(d_C-2)\sum_{n=1}^{d_C-2}w_{d_C-1}^{-(n-1)\ell}=-(d_C-2)w_{d_C-1}^{\ell}.
   \end{equation*}
   	\item If $s=\ell$, then
   \begin{equation*}
   \braket{\xi_\ell}{\eta_\ell}_C\braket{\eta_\ell}{\xi_0}_C=-\sum_{n=1}^{d_C-2}w_{d_C-1}^{\ell}=-(d_C-2)w_{d_C-1}^{\ell}.
   \end{equation*}
      \item If $s\neq 0,\ell$, then
   \begin{equation*}
   \braket{\xi_\ell}{\eta_s}_C\braket{\eta_s}{\xi_0}_C=-\sum_{n=1}^{d_C-2}w_{d_C-1}^{ns-(n-1)\ell}=-w_{d_C-1}^\ell\sum_{n=1}^{d_C-2}w_{d_C-1}^{(s-\ell)n}=w_{d_C-1}^\ell.
   \end{equation*}
\end{enumerate}
 Thus by Eq.~\eqref{eq:thezero}, we have
\begin{equation}\label{eq:a_0}
\sum_{s=0}^{d_C-2}a_s-(d_C-1)(a_0+a_\ell)+(d_B-2)(d_C-2)e_{0,0}=0.
\end{equation}
Since $\ell\neq 0\in\bbZ_{d_C-1}$, we must have $a_1=a_2=\cdots=a_{d_C-2}$. Then Eq.~\eqref{eq:a_0} can be expressed by
\begin{equation}\label{eq:a_1}
-(d_C-2)a_0-a_1+(d_B-2)(d_C-2)e_{0,0}=0.
\end{equation}
 Further, by using the states $\ket{S}$, $\ket{0}_A\ket{\eta_1}_B\ket{\xi_0}\in\cB_3$ and Eq.~\eqref{eq:operatorM} we obtain
\begin{equation*}
{}_B\left(\sum_{i_1=0}^{d_B-1}\bra{i_1}\right){}_C\left(\sum_{i_2=0}^{d_B-1}\bra{i_2}\right)E\ket{\eta_1}_B\ket{\xi_0}_C=(d_C-1)(d_C-2)a_0-(d_B-2)(d_C-2)^2e_{0,0}=0,
\end{equation*}
 i,e.
\begin{equation}\label{eq:S}
(d_C-1)a_0-(d_B-2)(d_C-2)e_{0,0}=0.
\end{equation}
Then, by Eqs.~\eqref{eq:a_1} and \eqref{eq:S}, it implies $a_0=a_1$. Thus $a_0=a_1=\ldots=a_{d_C-2}=k$. It means that the operator
\begin{equation*}
E_{\cA_1^{(A)}}=k\sum_{s=0}^{d_C-2}\ket{0}_B\bra{0}\otimes\ket{\eta_s}_C\bra{\eta_s},
\end{equation*}
 which is equivalent to
 \begin{equation}\label{eq:A1}
 E_{\cA_1^{(A)}}=k\bbI_{\cA_1^{(A)}}.
 \end{equation}

\noindent{\bf Step 3}
Considering $\ket{S}$ and  $\{\ket{0}_A\ket{\eta_i}_B\ket{\xi_j}_C\}_{(i,j) \in\bbZ_{d_B-1}\times\bbZ_{d_C-1}\setminus\{(0,0)\}}= \cB_3$. By using Eqs.~\eqref{eq:matrixE}  and \eqref{eq:A1}, we have
\begin{equation*}
{}_B\left(\sum_{i_1=0}^{d_B-1}\bra{i_1}\right){}_C\left(\sum_{i_2=0}^{d_C-1}\bra{i_2}\right)E\ket{\eta_i}_B\ket{\xi_j}_C={}_B\left(\sum_{i_1=0}^{d_B-2}\bra{i_1}\right){}_C\left(\sum_{i_2=1}^{d_C-1}\bra{i_2}\right)E\ket{\eta_i}_B\ket{\xi_j}_C=0.
\end{equation*}
Moreover, we have
\begin{equation*}
\left(\sum_{i_1=0}^{d_B-2}\ket{i_1}\right)_B\left(\sum_{i_2=1}^{d_C-1}\ket{i_2}\right)_C=\ket{\eta_0}_B\ket{\xi_0}_C.
\end{equation*}
Therefore, by using the states $\{\ket{S}\}\cup\{\ket{0}_A\ket{\eta_i}_B\ket{\xi_j}_C\}_{(i,j)\in\bbZ_{d_B-1}\times \bbZ_{d_C-1}\setminus\{(0,0)\}}$, we have
\begin{equation*}
{}_B\bra{\eta_k}_C\bra{\xi_\ell}E\ket{\eta_i}_B\ket{\xi_j}_C=0,  \quad  \text{for} \ (k,\ell)\neq (i,j)\in\bbZ_{d_B-1}\times \bbZ_{d_C-1}.
\end{equation*}
For any $\ket{j}_B\ket{k}_C\in \cA_1^{(A)}\cap\cB_3^{(A)}$, we have  ${}_{\{\ket{j}_B\ket{k}_C\}}E_{\cB_3^{(A)}\setminus \{\ket{j}_B\ket{k}_C\}}=\textbf{0}$ by Eqs.~\eqref{eq:matrixE}  and \eqref{eq:A1}.  Applying Block Trivial  Lemma to $\{\ket{\eta_i}_B\ket{\xi_j}_C\}_{(i,j)\in\bbZ_{d_B-1}\times \bbZ_{d_C-1}}$, we have
\begin{equation}\label{eq:B3}
E_{\cB_3^{(A)}}=k_1\bbI_{\cB_3^{(A)}}.
\end{equation} Since $\cA_1^{(A)}\cap \cB_3^{(A)}\neq \emptyset$, it implies $k=k_1$. Thus, by Eqs.~\eqref{eq:A1}  and \eqref{eq:B3}, we obtain
\begin{equation*}
E_{\cA_1^{(A)}\cup\cB_3^{(A)}}=k\bbI_{\cA_1^{(A)}\cup\cB_3^{(A)}}.
\end{equation*}
By the symmetry of Fig.~\ref{fig:d_Ad_Bd_C}, we can obtain  $E=k\bbI$. Thus,  $E$  is trivial. 	
\end{proof}
\vspace{0.4cm}

\section{The proof of Theorem~\ref{thm:d1d2d3}}\label{Appendix:thmd1d2d3}
\begin{proof}
First, we need to show that $\cU$ is a UPB. For the same discussion as Proposition~\ref{pro:3d1d2}, we can obtain the matrices,
\begin{equation*}
	M^{(d_A)}=\left[\begin{array}{ccccccccccccccccc}
	a_0 &a_0 &\cdots &a_0 &b_0 &\cdots &b_0 &c_0 &\cdots &c_0 &c_0 &\cdots &c_0 &c_0 &\cdots &c_0 &d_0\\
	a_0 &a_0 &\cdots &a_0 &b_0 &\cdots &b_0 & & & &b_1 &\cdots &b_1 &a_1&\cdots &a_1 &a_1\\
	\vdots &\vdots &\ddots  &\vdots&\vdots&\ddots&\vdots&&M^{(d_A-2)}&&\vdots&\ddots&\vdots&\vdots&\ddots&\vdots&\vdots\\
	a_0 &a_0 &\cdots &a_0 &b_0 &\cdots &b_0 & && &b_1 &\cdots &b_1 &a_1&\cdots &a_1 &a_1\\
	d_1 &c_1 &\cdots &c_1 &c_1 &\cdots &c_1 &c_1 &\cdots &c_1 &b_1 &\cdots &b_1 &a_1 &\ldots &a_1 &a_1
\end{array}	\right],
\end{equation*}
\begin{equation*}
M^{(d_A-2)}=\left[\begin{array}{ccccccccccccccccc}
a_0^{(1)} &a_0^{(1)} &\cdots &a_0^{(1)} &b_0^{(1)} &\cdots &b_0^{(1)} &c_0^{(1)} &\cdots &c_0^{(1)} &c_0^{(1)} &\cdots &c_0^{(1)} &c_0^{(1)} &\cdots &c_0^{(1)} &d_0^{(1)}\\
a_0^{(1)} &a_0^{(1)} &\cdots &a_0^{(1)} &b_0^{(1)} &\cdots &b_0^{(1)} & & & &b_1^{(1)} &\cdots &b_1^{(1)} &a_1^{(1)}&\cdots &a_1^{(1)} &a_1^{(1)}\\
\vdots &\vdots &\ddots  &\vdots&\vdots&\ddots&\vdots&&M^{(d_A-4)}&&\vdots&\ddots&\vdots&\vdots&\ddots&\vdots&\vdots\\
a_0^{(1)} &a_0^{(1)} &\cdots &a_0^{(1)} &b_0^{(1)} &\cdots &b_0^{(1)} & && &b_1^{(1)} & &b_1^{(1)} &a_1^{(1)}&\cdots &a_1^{(1)} &a_1^{(1)}\\
d_1^{(1)} &c_1^{(1)} &\cdots &c_1^{(1)} &c_1^{(1)} &\cdots &c_1^{(1)} &c_1^{(1)} &\cdots &c_1^{(1)} &b_1^{(1)} &\cdots &b_1^{(1)} &a_1^{(1)} &\ldots &a_1^{(1)} &a_1^{(1)}
\end{array}	\right],
\end{equation*}
\begin{equation*}
\vdots
\end{equation*}
\begin{equation*}
M^{(d_A-2s)}=\left[\begin{array}{ccccccccccccccccc}
a_0^{(s)} &a_0^{(s)} &\cdots &a_0^{(s)} &b_0^{(s)} &\cdots &b_0^{(s)} &c_0^{(s)} &\cdots &c_0^{(s)} &c_0^{(s)} &\cdots &c_0^{(s)} &c_0^{(s)} &\cdots &c_0^{(s)} &d_0^{(s)}\\
a_0^{(s)} &a_0^{(s)} &\cdots &a_0^{(s)} &b_0^{(s)} &\cdots &b_0^{(s)} &e^{(s)} &\cdots &e^{(s)} &b_1^{(s)} &\cdots &b_1^{(s)} &a_1^{(s)}&\cdots &a_1^{(s)} &a_1^{(s)}\\
\vdots &\vdots &\ddots  &\vdots&\vdots&\ddots&\vdots&\vdots&\ddots&\vdots&\vdots&\ddots&\vdots&\vdots&\ddots&\vdots&\vdots\\
a_0^{(s)} &a_0^{(s)} &\cdots &a_0^{(s)} &b_0^{(s)} &\cdots &b_0^{(s)} &e^{(s)} &\cdots&e^{(s)} &b_1^{(s)} & &b_1^{(s)} &a_1^{(s)}&\cdots &a_1^{(s)} &a_1^{(s)}\\
d_1^{(s)} &c_1^{(s)} &\cdots &c_1^{(s)} &c_1^{(s)} &\cdots &c_1^{(s)} &c_1^{(s)} &\cdots &c_1^{(s)} &b_1^{(s)} &\cdots &b_1^{(s)} &a_1^{(s)} &\ldots &a_1^{(s)} &a_1^{(s)}
\end{array}	\right],
\end{equation*}
where $M^{(d_A-2s)}$ is similar to $M$ of Eq.~\eqref{eq:M} in Proposition~\ref{pro:3d1d2}.  Next, we have

\begin{equation}\label{eq:rankd_1d_2d3}
\rank(M^{(d_A)})=1, \quad	\text{sum}(M^{(d_A)})=0.
\end{equation}
Then we consider $AB|C$ bipartition. We can rearrange  the first row of $M^{(d_A)}$ to the $d_B\times d_C$ matrix $M_{(d_A-1)}$  through $(AB,C)$ coordinates of $M^{(d_A)}$, and rearrange  the last row of $M^{(d_A)}$ to the $d_B\times d_C$ matrix $M_{0}$  through $(AB,C)$ coordinates of $M^{(d_A)}$, where
\begin{equation}
M_{(d_A-1)}=\begin{bmatrix}\label{eq:M_2d_1d_2d_3}
c_0 &c_0 &\cdots &c_0 &d_0\\
c_0 &c_0 &\cdots &c_0 &b_0\\
\vdots 	&\vdots  	&\ddots  	&\vdots 	&\vdots\\
c_0 &c_0 &\cdots &c_0 &b_0\\
a_0 &a_0 &\cdots &a_0 &b_0\\
\end{bmatrix}, \quad \rank(M_{(d_A-1)})=0\ \ \text{or} \ \  1,
\end{equation}
and
\begin{equation}
M_0=\begin{bmatrix}\label{eq:M_03_1d_2d_3}
b_1 &a_1 &\cdots &a_1 &a_1\\
b_1 &c_1 &\cdots &c_1 &c_1\\
\vdots 	&\vdots  	&\ddots  	&\vdots 	&\vdots\\
b_1 &c_1 &\cdots &c_1 &c_1\\
d_1 &c_1 &\cdots &c_1 &c_1\\
\end{bmatrix}, \quad \rank(M_0)=0\ \ \text{or} \ \  1.
\end{equation}
For the same proof of Example~\ref{example:334}, we can show that $a_0=a_1=b_0=b_1=c_0=c_1=d_0=d_1=0$ by Eqs.~\eqref{eq:rankd_1d_2d3}, \eqref{eq:M_2d_1d_2d_3}, and \eqref{eq:M_03_1d_2d_3}.
Then we obtain that
\begin{equation}\label{eq:rankd_1d_2d3^{d_A-2}}
\rank(M^{(d_A-2)})=1, \quad	\text{sum}(M^{(d_A-2)})=0.
\end{equation}
 We can rearrange  the second row of $M^{(d_A)}$ to the $d_B\times d_C$ matrix $M_{(d_A-2)}$  through $(AB,C)$ coordinates of $M^{(d_A)}$, and rearrange the last but two row of  $M^{(d_A)}$  to the $d_B\times d_C$ matrix $M_{1}$  through $(AB,C)$ coordinates of $M^{(d_A)}$, where
\begin{equation}
M_{(d_A-2)}=\begin{bmatrix}\label{eq:M_Ad_1d_2d_3}
0 &0   &0 & \cdots &0 &0 &0\\
0 &c_0^{(1)} &c_0^{(1)} &\cdots &c_0^{(1)} &d_0^{(1)} &0\\
0 &c_0^{(1)} &c_0^{(1)} &\cdots &c_0^{(1)} &b_0^{(1)} &0\\
\vdots&\vdots 	&\vdots  	&\ddots  	&\vdots 	&\vdots &\vdots\\
0 &c_0^{(1)} &c_0^{(1)} &\cdots &c_0^{(1)} &b_0^{(1)} &0\\
0 &a_0^{(1)} &a_0^{(1)} &\cdots &a_0^{(1)} &b_0^{(1)} &0\\
0 &0   &0   &\cdots &0   &0 &0
\end{bmatrix}, \quad \rank(M_{d_A-2})=0\ \ \text{or} \ \  1,
\end{equation}
and
\begin{equation}
M_1=\begin{bmatrix}\label{eq:M_A3_1d_2d_3}
0 &0   &0 & \cdots &0 &0 &0 \\
0&b_0^{(1)} &a_0^{(1)} &\cdots &a_0^{(1)} &a_0^{(1)} &0\\
0&b_0^{(1)} &c_0^{(1)} &\cdots &c_0^{(1)} &c_0^{(1)} &0\\
\vdots&\vdots 	&\vdots  	&\ddots  	&\vdots 	&\vdots &\vdots\\
0&b_0^{(1)} &c_0^{(1)} &\cdots &c_0^{(1)} &c_0^{(1)} &0\\
0&d_0^{(1)} &c_0^{(1)} &\cdots &c_0^{(1)} &c_0^{(1)} &0\\
0 &0   &0   &\cdots &0   &0 &0
\end{bmatrix}, \quad \rank(M_1)=0\ \ \text{or} \ \  1.
\end{equation}
Similarly, we can also show that $a_0^{(1)}=a_1^{(1)}=b_0^{(1)}=b_1^{(1)}=c_0^{(1)}=c_1^{(1)}=d_0^{(1)}=d_1^{(1)}=0$ by Eqs.~\eqref{eq:rankd_1d_2d3^{d_A-2}}, \eqref{eq:M_Ad_1d_2d_3}, and \eqref{eq:M_A3_1d_2d_3}.
Repeating this process $s$ times, then we obtain that
\begin{equation}\label{eq:rankd_1d_2d3^{d_A-s}}
\rank(M^{(d_A-2s)})=1, \quad	\text{sum}(M^{(d_A-2s)})=0.
\end{equation}
Since $M^{(d_A-2s)}$ is similar to $M$ of Eq.~\eqref{eq:M} in Proposition~\ref{pro:3d1d2}, and we can show that $a_0^{(s)}=a_1^{(s)}=b_0^{(s)}=b_1^{(s)}=c_0^{(s)}=c_1^{(s)}=d_0^{(s)}=d_1^{(s)}=e^{(s)}=0$ by the proof of Proposition~\ref{pro:3d1d2}. Thus we obtain that $M^{(d_A)}$ is a zero matrix, and it is impossible for $\rank(M^{(d_A)}) =1$.

We can obtain that  $\cU$ is strongly nonlocal by induction on $t$ along with Proposition~\ref{pro:dAdBdCofthe}.
\end{proof}
\vspace{0.4cm}

\section{The proof of Proposition~\ref{pro:dis}}\label{Appendix:prodis}
\begin{proof}
      We denote $\ket{\Phi(2)}_{A,B}$ as $\ket{\Phi(2)}_{a,b_1}$, and $\ket{\Phi(2)}_{B,C}$ as  $\ket{\Phi(2)}_{b_2,c}$. Assume that $\ket{\Phi(2)}_{a,b_1}$ is distributed between Alice and Bob, and $\ket{\Phi(2)}_{b_2,c}$ is distributed between Bob and Charlie. The initial states are
	\begin{equation}\label{eq:states}
	\ket{\psi}_{A,B,C}(\ket{0}_a\ket{0}_{b_1}+\ket{1}_a\ket{1}_{b_1})(\ket{0}_{b_2}\ket{0}_c+\ket{1}_{b_2}\ket{1}_c),
	\end{equation}
	for $\ket{\psi}_{A,B,C}\in\cup_{i=1}^3(\cA_i\cup\cB_i)\cup\cF\cup \{\ket{S}\}$,
	where $a$ is the ancillary system of Alice, $b_1$ and $b_2$ are the ancillary systems of Bob, and $c$ is the ancillary system of Charlie. Denote  $P[\ket{i}_{\spadesuit}]:=\ketbra{i}{i}_\spadesuit$,
	$P[(\ket{i_1},\ket{i_2},\ldots,i_r)_\spadesuit;(\ket{j_1},\ket{j_2},\ldots,\ket{j_s})_\clubsuit;(\ket{k_1},\ket{k_2},\ldots,\ket{k_t})_\diamondsuit]:=(\ketbra{i_1}{i_1}+\ketbra{i_2}{i_2}+\cdots+\ketbra{i_r}{i_r})_\spadesuit\otimes (\ketbra{j_1}{j_1}+\ketbra{j_2}{j_2}+\cdots+\ketbra{j_s}{j_s})_\clubsuit\otimes(\ketbra{k_1}{k_1}+\ketbra{k_2}{k_2}+\cdots+\ketbra{k_t}{k_t})_\diamondsuit$.
    Now the discrimination protocol proceeds as follows.
	
	\textit{Step 1}.  Alice performs the measurement $\{M_1:=P[(\ket{0},\ket{1})_A;\ket{0}_a]+P[\ket{2}_A;\ket{1}_a], \overline{M_1}:=\bbI-M_1\}$.  Charlie performs the measurement $\{L_1:=P[(\ket{1},\ket{2},\ket{3})_C;\ket{0}_c]+P[\ket{0}_C;\ket{1}_c], \overline{L_1}:=\bbI-L_1\}$.  If $M_1$ and $L_1$ clicks (it means that the operators $M_1$ and $L_1$ act on the states in Eq.~\eqref{eq:states}), the resulting postmeasurement states are	
	\begin{equation}\label{eq:post}
	\begin{aligned}
	\ket{\psi_{1}(0,1)}\rightarrow&\ket{1}_A\ket{0}_B\ket{0}_C\ket{0}_a\ket{0}_{b_1}\ket{1}_{b_2}\ket{1}_c+\ket{2}_A\ket{0}_B\ket{0}_C\ket{1}_a\ket{1}_{b_1}\ket{1}_{b_2}\ket{1}_c\\
	&+\ket{1}_A\ket{0}_B(w_3\ket{1}+w_3^2\ket{2})_C\ket{0}_a\ket{0}_{b_1}\ket{0}_{b_2}\ket{0}_c+\ket{2}_A\ket{0}_B(w_3\ket{1}+w_3^2\ket{2})_C\ket{1}_a\ket{1}_{b_1}\ket{0}_{b_2}\ket{0}_c,\\
	\ket{\psi_{1}(0,2)}\rightarrow&\ket{1}_A\ket{0}_B\ket{0}_C\ket{0}_a\ket{0}_{b_1}\ket{1}_{b_2}\ket{1}_c+\ket{2}_A\ket{0}_B\ket{0}_C\ket{1}_a\ket{1}_{b_1}\ket{1}_{b_2}\ket{1}_c\\
	&+\ket{1}_A\ket{0}_B(w_3^2\ket{1}+w_3\ket{2})_C\ket{0}_a\ket{0}_{b_1}\ket{0}_{b_2}\ket{0}_c+\ket{2}_A\ket{0}_B(w_3^2\ket{1}+w_3\ket{2})_C\ket{1}_a\ket{1}_{b_1}\ket{0}_{b_2}\ket{0}_c,\\
	\ket{\psi_{1}(1,0)}\rightarrow&\ket{1}_A\ket{0}_B\ket{0}_C\ket{0}_a\ket{0}_{b_1}\ket{1}_{b_2}\ket{1}_c-\ket{2}_A\ket{0}_B\ket{0}_C\ket{1}_a\ket{1}_{b_1}\ket{1}_{b_2}\ket{1}_c\\
    &+\ket{1}_A\ket{0}_B(\ket{1}+\ket{2})_C\ket{0}_a\ket{0}_{b_1}\ket{0}_{b_2}\ket{0}_c-\ket{2}_A\ket{0}_B(\ket{1}+\ket{2})_C\ket{1}_a\ket{1}_{b_1}\ket{0}_{b_2}\ket{0}_c,\\
	\ket{\psi_{1}(1,1)}\rightarrow&\ket{1}_A\ket{0}_B\ket{0}_C\ket{0}_a\ket{0}_{b_1}\ket{1}_{b_2}\ket{1}_c-\ket{2}_A\ket{0}_B\ket{0}_C\ket{1}_a\ket{1}_{b_1}\ket{1}_{b_2}\ket{1}_c\\
    &+\ket{1}_A\ket{0}_B(w_3\ket{1}+w_3^2\ket{2})_C\ket{0}_a\ket{0}_{b_1}\ket{0}_{b_2}\ket{0}_c-\ket{2}_A\ket{0}_B(w_3\ket{1}+w_3^2\ket{2})_C\ket{1}_a\ket{1}_{b_1}\ket{0}_{b_2}\ket{0}_c,\\
	\ket{\psi_{1}(1,2)}\rightarrow&\ket{1}_A\ket{0}_B\ket{0}_C\ket{0}_a\ket{0}_{b_1}\ket{1}_{b_2}\ket{1}_c-\ket{2}_A\ket{0}_B\ket{0}_C\ket{1}_a\ket{1}_{b_1}\ket{1}_{b_2}\ket{1}_c\\
    &+\ket{1}_A\ket{0}_B(w_3^2\ket{1}+w_3\ket{2})_C\ket{0}_a\ket{0}_{b_1}\ket{0}_{b_2}\ket{0}_c-\ket{2}_A\ket{0}_B(w_3^2\ket{1}+w_3\ket{2})_C\ket{1}_a\ket{1}_{b_1}\ket{0}_{b_2}\ket{0}_c,\\
	\ket{\psi_{2}(0,1)}\rightarrow&\ket{1}_A(\ket{0}-\ket{1})_B\ket{3}_C\ket{0}_a\ket{0}_{b_1}\ket{0}_{b_2}\ket{0}_c+\ket{2}_A(\ket{0}-\ket{1})_B\ket{3}_C\ket{1}_a\ket{1}_{b_1}\ket{0}_{b_2}\ket{0}_c,\\
	\ket{\psi_{2}(1,0)}\rightarrow&\ket{1}_A(\ket{0}+\ket{1})_B\ket{3}_C\ket{0}_a\ket{0}_{b_1}\ket{0}_{b_2}\ket{0}_c-\ket{2}_A(\ket{0}+\ket{1})_B\ket{3}_C\ket{1}_a\ket{1}_{b_1}\ket{0}_{b_2}\ket{0}_c,\\
	\ket{\psi_{2}(1,1)}\rightarrow&\ket{1}_A(\ket{0}-\ket{1})_B\ket{3}_C\ket{0}_a\ket{0}_{b_1}\ket{0}_{b_2}\ket{0}_c-\ket{2}_A(\ket{0}-\ket{1})_B\ket{3}_C\ket{1}_a\ket{1}_{b_1}\ket{0}_{b_2}\ket{0}_c,\\
	\ket{\psi_{3}(0,1)}\rightarrow&\ket{2}_A(\ket{1}+\ket{2})_B\ket{0}_C\ket{1}_a\ket{1}_{b_1}\ket{1}_{b_2}\ket{1}_c+\ket{2}_A(\ket{1}+\ket{2})_B(w_3\ket{1}+w_3^2\ket{2})_C\ket{1}_a\ket{1}_{b_1}\ket{0}_{b_2}\ket{0}_c,\\
	\ket{\psi_{3}(0,2)}\rightarrow&\ket{2}_A(\ket{1}+\ket{2})_B\ket{0}_C\ket{1}_a\ket{1}_{b_1}\ket{1}_{b_2}\ket{1}_c+\ket{2}_A(\ket{1}+\ket{2})_B(w_3^2\ket{1}+w_3\ket{2})_C\ket{1}_a\ket{1}_{b_1}\ket{0}_{b_2}\ket{0}_c,\\
	\ket{\psi_{3}(1,0)}\rightarrow&\ket{2}_A(\ket{1}-\ket{2})_B\ket{0}_C\ket{1}_a\ket{1}_{b_1}\ket{1}_{b_2}\ket{1}_c+\ket{2}_A(\ket{1}-\ket{2})_B(\ket{1}+\ket{2})_C\ket{1}_a\ket{1}_{b_1}\ket{0}_{b_2}\ket{0}_c,\\
	\ket{\psi_{3}(1,1)}\rightarrow&\ket{2}_A(\ket{1}-\ket{2})_B\ket{0}_C\ket{1}_a\ket{1}_{b_1}\ket{1}_{b_2}\ket{1}_c+\ket{2}_A(\ket{1}-\ket{2})_B(w_3\ket{1}+w_3^2\ket{2})_C\ket{1}_a\ket{1}_{b_1}\ket{0}_{b_2}\ket{0}_c,
	\end{aligned}
	\end{equation}
	\begin{equation*}
	\begin{aligned}
	\ket{\psi_{3}(1,2)}\rightarrow&\ket{2}_A(\ket{1}-\ket{2})_B\ket{0}_C\ket{1}_a\ket{1}_{b_1}\ket{1}_{b_2}\ket{1}_c+\ket{2}_A(\ket{1}-\ket{2})_B(w_3^2\ket{1}+w_3\ket{2})_C\ket{1}_a\ket{1}_{b_1}\ket{0}_{b_2}\ket{0}_c,\\
	\ket{\phi_{1}(0,1)}\rightarrow&(\ket{0}+\ket{1})_A\ket{2}_B(\ket{1}+w_3\ket{2}+w_3^2\ket{3})_C\ket{0}_a\ket{0}_{b_1}\ket{0}_{b_2}\ket{0}_c,\\
	\ket{\phi_{1}(0,2)}\rightarrow&(\ket{0}+\ket{1})_A\ket{2}_B(\ket{1}+w_3^2\ket{2}+w_3\ket{3})_C\ket{0}_a\ket{0}_{b_1}\ket{0}_{b_2}\ket{0}_c,\\
	\ket{\phi_{1}(1,0)}\rightarrow&(\ket{0}-\ket{1})_A\ket{2}_B(\ket{1}+\ket{2}+\ket{3})_C\ket{0}_a\ket{0}_{b_1}\ket{0}_{b_2}\ket{0}_c,\\
	\ket{\phi_{1}(1,1)}\rightarrow&(\ket{0}-\ket{1})_A\ket{2}_B(\ket{1}+w_3\ket{2}+w_3^2\ket{3})_C\ket{0}_a\ket{0}_{b_1}\ket{0}_{b_2}\ket{0}_c,\\
	\ket{\phi_{1}(1,2)}\rightarrow&(\ket{0}-\ket{1})_A\ket{2}_B(\ket{1}+w_3^2\ket{2}+w_3\ket{3})_C\ket{0}_a\ket{0}_{b_1}\ket{0}_{b_2}\ket{0}_c,\\
	\ket{\phi_{2}(0,1)}\rightarrow&(\ket{0}+\ket{1})_A(\ket{1}-\ket{2})_B\ket{0}_C\ket{0}_a\ket{0}_{b_1}\ket{1}_{b_2}\ket{1}_c,\\
	\ket{\phi_{2}(1,0)}\rightarrow&(\ket{0}-\ket{1})_A(\ket{1}+\ket{2})_B\ket{0}_C\ket{0}_a\ket{0}_{b_1}\ket{1}_{b_2}\ket{1}_c,\\
	\ket{\phi_{2}(1,1)}\rightarrow&(\ket{0}-\ket{1})_A(\ket{1}-\ket{2})_B\ket{0}_C\ket{0}_a\ket{0}_{b_1}\ket{1}_{b_2}\ket{1}_c,\\
	\ket{\phi_3(0,1)}\rightarrow&\ket{0}_A(\ket{0}+\ket{1})_B(\ket{1}+w_3\ket{2}+w_3^2\ket{3})_C\ket{0}_a\ket{0}_{b_1}\ket{0}_{b_2}\ket{0}_c,\\
	\ket{\phi_3(0,2)}\rightarrow&\ket{0}_A(\ket{0}+\ket{1})_B(\ket{1}+w_3^2\ket{2}+w_3\ket{3})_C\ket{0}_a\ket{0}_{b_1}\ket{0}_{b_2}\ket{0}_c,\\
	\ket{\phi_3(1,0)}\rightarrow&\ket{0}_A(\ket{0}-\ket{1})_B(\ket{1}+\ket{2}+\ket{3})_C\ket{0}_a\ket{0}_{b_1}\ket{0}_{b_2}\ket{0}_c,\\
	\ket{\phi_3(1,1)}\rightarrow&\ket{0}_A(\ket{0}-\ket{1})_B(\ket{1}+w_3\ket{2}+w_3^2\ket{3})_C\ket{0}_a\ket{0}_{b_1}\ket{0}_{b_2}\ket{0}_c,\\
	\ket{\phi_3(1,2)}\rightarrow&\ket{0}_A(\ket{0}-\ket{1})_B(\ket{1}+w_3^2\ket{2}+w_3\ket{3})_C\ket{0}_a\ket{0}_{b_1}\ket{0}_{b_2}\ket{0}_c,\\
	\ket{\varphi(1)}\rightarrow&\ket{1}_A\ket{1}_B(\ket{1}-\ket{2})_C\ket{0}_a\ket{0}_{b_1}\ket{0}_{b_2}\ket{0}_c,\\
	\ket{S}\rightarrow&(\ket{0}+\ket{1})_A(\ket{0}+\ket{1}+\ket{2})_B\ket{0}_C\ket{0}_a\ket{0}_{b_1}\ket{1}_{b_2}\ket{1}_c+\ket{2}_A(\ket{0}+\ket{1}+\ket{2})_B\ket{0}_C\ket{1}_a\ket{1}_{b_1}\ket{1}_{b_2}\ket{1}_c+\\
	&(\ket{0}+\ket{1})_A(\ket{0}+\ket{1}+\ket{2})_B(\ket{1}+\ket{2}+\ket{3})_C\ket{0}_a\ket{0}_{b_1}\ket{0}_{b_2}\ket{0}_c\\+&\ket{2}_A(\ket{0}+\ket{1}+\ket{2})_B(\ket{1}+\ket{2}+\ket{3})_C\ket{1}_a\ket{1}_{b_1}\ket{0}_{b_2}\ket{0}_c.
	\end{aligned}
	\end{equation*}	
	
	\textit{Step 2}. Bob performs the measurement $\{M_{2,1}:=P[\ket{2}_B;\ket{0}_{b_1};\ket{0}_{b_2}], M_{2,2}:=P[(\ket{1}-\ket{2})_B;\ket{0}_{b_1};\ket{1}_{b_2}],M_{2,3}:=P[(\ket{1}+\ket{2})_B;\ket{0}_{b_1};\ket{1}_{b_2}], \overline{M_2}:=\bbI-\sum_{j=1}^3M_{2,j}\}$. If $M_{2,1}$ clicks (it means that the operator $M_{2,1}$ acts on the states in Eq.~\eqref{eq:post}), then the postmeasurement states are $\{\ket{\phi_1(i,j)}\}_{(i,j)\neq(0,0)\in\bbZ_2\times\bbZ_3}$, and $\ket{S}\rightarrow (\ket{0}+\ket{1})_A\ket{2}_B(\ket{1}+\ket{2}+\ket{3})_C\ket{0}_a\ket{0}_{b_1}\ket{0}_{b_2}\ket{0}_c$. Then Alice performs the measurement  $\{M_{2,1,1}:=P[(\ket{0}-\ket{1})_A], \overline{M_{2,1,1}}=\bbI-M_{2,1,1}\}$. If $M_{2,1,1}$ clicks (it means that the operator $M_{2,1,1}$ acts on the states $\{\ket{\phi_1(i,j)}\}_{(i,j)\neq(0,0)\in\bbZ_2\times\bbZ_3}$ and $\ket{S}$), then the postmeasurement states are $\{\ket{\phi_{1}(1,j)}\}_{j\in\bbZ_3}$, which are locally distinguishable; if $\overline{M_{2,1,1}}$ clicks, then the postmeasurement states are $\{\ket{\phi_{1}(0,j)}\}_{j\neq 0\in\bbZ_3}$ and $\ket{S}$ that are locally distinguishable. Next, if $M_{2,2}$ clicks, then the postmeasurement states are $\{\ket{\phi_{2}(i,1)}\}_{i\in\bbZ_2}$ that are locally distinguishable; if $M_{2,3}$ clicks, then the postmeasurement states are $\ket{\phi_2(1,0)}$ and  $\ket{S}\rightarrow(\ket{0}+\ket{1})_A(\ket{1}+\ket{2})_B\ket{0}_C\ket{0}_a\ket{0}_{b_1}\ket{1}_{b_2}\ket{1}_c$  that are locally distinguishable;  if $\overline{M_2}$ clicks, then the postmeasurement states are $\{\ket{\psi_1(i,j)}\}_{(i,j)\neq(0,0)\in\bbZ_2\times\bbZ_3}$, $\{\ket{\psi_2(i,j)}\}_{(i,j)\neq(0,0)\in\bbZ_2\times\bbZ_2}$, $\{\ket{\psi_3(i,j)}\}_{(i,j)\neq(0,0)\in\bbZ_2\times\bbZ_3}$, $\{\ket{\phi_3(i,j)}\}_{(i,j)\neq(0,0)\in\bbZ_2\times\bbZ_3}$, $\ket{\varphi(1)}$ and $\ket{S}\rightarrow(\ket{0}+\ket{1})_A\ket{0}_B\ket{0}_C\ket{0}_a\ket{0}_{b_1}\ket{1}_{b_2}\ket{1}_c+\ket{2}_A(\ket{0}+\ket{1}+\ket{2})_B\ket{0}_C\ket{1}_a\ket{1}_{b_1}\ket{1}_{b_2}\ket{1}_c+
	(\ket{0}+\ket{1})_A(\ket{0}+\ket{1})_B(\ket{1}+\ket{2}+\ket{3})_C\ket{0}_a\ket{0}_{b_1}\ket{0}_{b_2}\ket{0}_c+\ket{2}_A(\ket{0}+\ket{1}+\ket{2})_B(\ket{1}+\ket{2}+\ket{3})_C\ket{1}_a\ket{1}_{b_1}\ket{0}_{b_2}\ket{0}_c$.

	\textit{Step 3}. Alice performs the measurement $\{M_3:=P[\ket{0}_A],\overline{M_3}=\bbI-M_3\}$. If $M_3$ clicks (it means that the operator $M_3$ acts on the states $\{\ket{\psi_1(i,j)}\}_{(i,j)\neq(0,0)\in\bbZ_2\times\bbZ_3}$, $\{\ket{\psi_2(i,j)}\}_{(i,j)\neq(0,0)\in\bbZ_2\times\bbZ_2}$, $\{\ket{\psi_3(i,j)}\}_{(i,j)\neq(0,0)\in\bbZ_2\times\bbZ_3}$, $\{\ket{\phi_3(i,j)}\}_{(i,j)\neq(0,0)\in\bbZ_2\times\bbZ_3}$, $\ket{\varphi(1)}$ and $\ket{S}$), then the postmeasurement states are $\{\ket{\phi_3(i,j)}\}_{(i,j)\neq(0,0)\in\bbZ_2\times\bbZ_3}$ and $\ket{S}\rightarrow \ket{0}_A\ket{0}_B\ket{0}_C\ket{0}_a\ket{0}_{b_1}\ket{1}_{b_2}\ket{1}_c+
	\ket{0}_A(\ket{0}+\ket{1})_B(\ket{1}+\ket{2}+\ket{3})_C\ket{0}_a\ket{0}_{b_1}\ket{0}_{b_2}\ket{0}_c$, which are locally distinguishable; if $\overline{M_3}$ clicks, then the postmeasurement states are  $\{\ket{\psi_1(i,j)}\}_{(i,j)\neq(0,0)\in\bbZ_2\times\bbZ_3}$, $\{\psi_2(i,j)\}_{(i,j)\neq(0,0)\in\bbZ_2\times\bbZ_2}$, $\{\ket{\psi_3(i,j)}\}_{(i,j)\neq(0,0)\in\bbZ_2\times\bbZ_3}$, $\ket{\varphi(1)}$, and $\ket{S}\rightarrow \ket{1}_A\ket{0}_B\ket{0}_C\ket{0}_a\ket{0}_{b_1}\ket{1}_{b_2}\ket{1}_c+\ket{2}_A(\ket{0}+\ket{1}+\ket{2})_B\ket{0}_C\ket{1}_a\ket{1}_{b_1}\ket{1}_{b_2}\ket{1}_c+
	\ket{1}_A(\ket{0}+\ket{1})_B(\ket{1}+\ket{2}+\ket{3})_C\ket{0}_a\ket{0}_{b_1}\ket{0}_{b_2}\ket{0}_c+\ket{2}_A(\ket{0}+\ket{1}+\ket{2})_B(\ket{1}+\ket{2}+\ket{3})_C\ket{1}_a\ket{1}_{b_1}\ket{0}_{b_2}\ket{0}_c$.
	
	\textit{Step 4}. Charlie performs the measurement $\{M_4:=P[\ket{3}_C],\overline{M_4}=\bbI-M_4\}$. If $M_4$ clicks (it means that the operator $M_4$ acts on the states $\{\ket{\psi_1(i,j)}\}_{(i,j)\neq(0,0)\in\bbZ_2\times\bbZ_3}$, $\{\ket{\psi_2(i,j)}\}_{(i,j)\neq(0,0)\in\bbZ_2\times\bbZ_2}$, $\{\ket{\psi_3(i,j)}\}_{(i,j)\neq(0,0)\in\bbZ_2\times\bbZ_3}$, $\ket{\varphi(1)}$ and $\ket{S}$), then then the postmeasurement states are $\{\psi_2(i,j)\}_{(i,j)\neq(0,0)\in\bbZ_2\times\bbZ_2}$, and  $\ket{S}\rightarrow
	\ket{1}_A(\ket{0}+\ket{1})_B\ket{3}_C\ket{0}_a\ket{0}_{b_1}\ket{0}_{b_2}\ket{0}_c+\ket{2}_A(\ket{0}+\ket{1}+\ket{2})_B\ket{3}_C\ket{1}_a\ket{1}_{b_1}\ket{0}_{b_2}\ket{0}_c$, which are locally distinguishable; if $\overline{M_4}$ clicks, then the postmeasurement states are  $\{\ket{\psi_1(i,j)}\}_{(i,j)\neq(0,0)\in\bbZ_2\times\bbZ_3}$, $\{\ket{\psi_3(i,j)}\}_{(i,j)\neq(0,0)\in\bbZ_2\times\bbZ_3}$, $\ket{\varphi(1)}$ and $\ket{S}\rightarrow \ket{1}_A\ket{0}_B\ket{0}_C\ket{0}_a\ket{0}_{b_1}\ket{1}_{b_2}\ket{1}_c+\ket{2}_A(\ket{0}+\ket{1}+\ket{2})_B\ket{0}_C\ket{1}_a\ket{1}_{b_1}\ket{1}_{b_2}\ket{1}_c+
	\ket{1}_A(\ket{0}+\ket{1})_B(\ket{1}+\ket{2})_C\ket{0}_a\ket{0}_{b_1}\ket{0}_{b_2}\ket{0}_c+\ket{2}_A(\ket{0}+\ket{1}+\ket{2})_B(\ket{1}+\ket{2})_C\ket{1}_a\ket{1}_{b_1}\ket{0}_{b_2}\ket{0}_c$.
	
	\textit{Step 5}. Bob performs the measurement $\{M_5:=P[\ket{0}_B],\overline{M_5}=\bbI-M_5\}$. If $M_5$ clicks (it means that the operator $M_5$ acts on the states $\{\ket{\psi_1(i,j)}\}_{(i,j)\neq(0,0)\in\bbZ_2\times\bbZ_3}$, $\{\ket{\psi_3(i,j)}\}_{(i,j)\neq(0,0)\in\bbZ_2\times\bbZ_3}$, $\ket{\varphi(1)}$ and $\ket{S}$), then the postmeasurement states are  $\{\ket{\psi_1(i,j)}\}_{(i,j)\neq(0,0)\in\bbZ_2\times\bbZ_3}$ and $\ket{S}\rightarrow \ket{1}_A\ket{0}_B\ket{0}_C\ket{0}_a\ket{0}_{b_1}\ket{1}_{b_2}\ket{1}_c+\ket{2}_A\ket{0}_B\ket{0}_C\ket{1}_a\ket{1}_{b_1}\ket{1}_{b_2}\ket{1}_c+
	\ket{1}_A\ket{0}_B(\ket{1}+\ket{2})_C\ket{0}_a\ket{0}_{b_1}\ket{0}_{b_2}\ket{0}_c+\ket{2}_A\ket{0}_B(\ket{1}+\ket{2})_C\ket{1}_a\ket{1}_{b_1}\ket{0}_{b_2}\ket{0}_c$. Then, Alice performs the measurement $\{M_{5,1}:=P[(\ket{0}+\ket{1})_a],\overline{M_{5,1}}=P[(\ket{0}-\ket{1})_a]\}$. Bob performs the measurement $\{M_{5,2}:=P[(\ket{0}+\ket{1})_{b_1}],\overline{M_{5,2}}=P[(\ket{0}-\ket{1})_{b_1}]\}$. If $M_{5,1}$ and $M_{5,2}$ clicks (It means that the operators $M_{5,1}$ and $M_{5,2}$ act on the states $\{\ket{\psi_1(i,j)}\}_{(i,j)\neq(0,0)\in\bbZ_2\times\bbZ_3}$ and $\ket{S}$ ), then the postmeasurement states are
	\begin{equation*}
	\begin{aligned}
	\ket{\psi_{1}(0,1)}\rightarrow&(\ket{1}+\ket{2})_A\ket{0}_B\ket{0}_C(\ket{0}+\ket{1})_a(\ket{0}+\ket{1})_{b_1}\ket{1}_{b_2}\ket{1}_c\\
	&+(\ket{1}+\ket{2})_A\ket{0}_B(w_3\ket{1}+w_3^2\ket{2})_C(\ket{0}+\ket{1})_a(\ket{0}+\ket{1})_{b_1}\ket{0}_{b_2}\ket{0}_c,\\
	\ket{\psi_{1}(0,2)}\rightarrow&(\ket{1}+\ket{2})_A\ket{0}_B\ket{0}_C(\ket{0}+\ket{1})_a(\ket{0}+\ket{1})_{b_1}\ket{1}_{b_2}\ket{1}_c\\
    &+(\ket{1}+\ket{2})_A\ket{0}_B(w_3^2\ket{1}+w_3\ket{2})_C(\ket{0}+\ket{1})_a(\ket{0}+\ket{1})_{b_1}\ket{0}_{b_2}\ket{0}_c,\\
	\ket{\psi_{1}(1,0)}\rightarrow&(\ket{1}-\ket{2})_A\ket{0}_B\ket{0}_C(\ket{0}+\ket{1})_a(\ket{0}+\ket{1})_{b_1}\ket{1}_{b_2}\ket{1}_c\\
    &+(\ket{1}-\ket{2})_A\ket{0}_B(\ket{1}+\ket{2})_C(\ket{0}+\ket{1})_a(\ket{0}+\ket{1})_{b_1}\ket{0}_{b_2}\ket{0}_c,\\
   	\ket{\psi_{1}(1,1)}\rightarrow&(\ket{1}-\ket{2})_A\ket{0}_B\ket{0}_C(\ket{0}+\ket{1})_a(\ket{0}+\ket{1})_{b_1}\ket{1}_{b_2}\ket{1}_c\\
    &+(\ket{1}-\ket{2})_A\ket{0}_B(w_3\ket{1}+w_3^2\ket{2})_C(\ket{0}+\ket{1})_a(\ket{0}+\ket{1})_{b_1}\ket{0}_{b_2}\ket{0}_c,\\
   	\ket{\psi_{1}(1,2)}\rightarrow&(\ket{1}-\ket{2})_A\ket{0}_B\ket{0}_C(\ket{0}+\ket{1})_a(\ket{0}+\ket{1})_{b_1}\ket{1}_{b_2}\ket{1}_c\\
    &+(\ket{1}-\ket{2})_A\ket{0}_B(w_3^2\ket{1}+w_3\ket{2})_C(\ket{0}+\ket{1})_a(\ket{0}+\ket{1})_{b_1}\ket{0}_{b_2}\ket{0}_c,\\
	\ket{S}\rightarrow&(\ket{1}+\ket{2})_A\ket{0}_B\ket{0}_C(\ket{0}+\ket{1})_a(\ket{0}+\ket{1})_{b_1}\ket{1}_{b_2}\ket{1}_c\\&+(\ket{1}+\ket{2})_A\ket{0}_B(\ket{1}+\ket{2})_C(\ket{0}+\ket{1})_a(\ket{0}+\ket{1})_{b_1}\ket{0}_{b_2}\ket{0}_c,\\
	\end{aligned}
	\end{equation*}
	which can be easily locally distinguished. All other cases obtain a similar protocol.  If $\overline{M_5}$ clicks, then the postmeasurement states are $\{\ket{\psi_3(i,j)}\}_{(i,j)\neq(0,0)\in\bbZ_2\times\bbZ_3}$, $\ket{\varphi(1)}$ and $\ket{S}\rightarrow \ket{2}_A(\ket{1}+\ket{2})_B\ket{0}_C\ket{1}_a\ket{1}_{b_1}\ket{1}_{b_2}\ket{1}_c+
	\ket{1}_A\ket{1}_B(\ket{1}+\ket{2})_C\ket{0}_a\ket{0}_{b_1}\ket{0}_{b_2}\ket{0}_c+\ket{2}_A(\ket{1}+\ket{2})_B(\ket{1}+\ket{2})_C\ket{1}_a\ket{1}_{b_1}\ket{0}_{b_2}\ket{0}_c$.
	
	\textit{Step 6}.  Alice performs the measurement $\{M_6:=P[\ket{0}_A],\overline{M_6}=\bbI-M_6\}$. If $M_6$ clicks (it means that the operator $M_6$ acts on the states $\{\ket{\psi_3(i,j)}\}_{(i,j)\neq(0,0)\in\bbZ_2\times\bbZ_3}$, $\ket{\varphi(1)}$ and $\ket{S}$), then the postmeasurement states are $\{\psi_3(i,j)\}_{(i,j)\neq(0,0)\in\bbZ_2\times\bbZ_3}$ and $\ket{S}$ that are locally distinguishable. If $\overline{M_6}$ clicks,  then the postmeasurement state is $\ket{\varphi(1)}$.
	
	All other cases in step 1 obtain a similar protocol. This completes the proof.
\end{proof}
\vspace{0.4cm}


\end{document}